\documentclass[11pt]{article}

\usepackage[top=1in, bottom=1in, left=1in, right=1in]{geometry}

\usepackage[american]{babel}
\usepackage{paralist}
\usepackage{wrapfig}

\usepackage{amsthm}
\usepackage{amsmath,amsfonts}
\usepackage{amssymb,bm,bbm}

\PassOptionsToPackage{hyphens}{url}
\usepackage[linktocpage,
            colorlinks=true,linkcolor=black,
            citecolor=black,urlcolor=black,
            pdfpagemode=UseNone,
            bookmarks=true,
            breaklinks,
            bookmarksopen=false,
            hyperfootnotes=false,
            pdfhighlight=/I,
            pdftitle={Improved approximation for Fréchet distance on c-packed curves matching conditional lower bounds},
            pdfauthor={Karl Bringmann, Marvin Künnemann}
            ]{hyperref}

\usepackage{thmtools, thm-restate}
\declaretheoremstyle[
	headfont=\normalfont\bfseries,
	notefont=\normalfont\bfseries,
	bodyfont=\normalfont
]{problemstyle}
\declaretheorem[numberwithin=section]{theorem}

\declaretheorem[sibling=theorem]{lemma}
\declaretheorem[sibling=theorem,name=Lemma]{lem}
\declaretheorem[sibling=theorem]{claim}
\declaretheorem[sibling=theorem]{definition}
\declaretheorem[sibling=theorem,style=problemstyle]{problem}
\declaretheorem[sibling=theorem,style=problemstyle,name=Algorithm]{algenv}
\declaretheorem[sibling=theorem,name=Proposition]{prop}

\usepackage{xspace}
\usepackage[latin1]{inputenc}
\usepackage{url}
\usepackage{tikz}
\usepackage{pgfplots,pgfplotstable,filecontents}
\usepackage{subfig}

\pgfplotsset{compat=newest}

\usepackage{float}
\floatstyle{ruled}
\newfloat{myalgorithm}{thp}{lop}
\floatname{myalgorithm}{Algorithm}

\usepackage{hyperref}
\usepackage{algorithm}
\usepackage[noend]{algpseudocode}
\usepackage{gensymb}

\usepackage[numbers,sort&compress,sectionbib]{natbib}
\usepackage[labelfont=bf,font=footnotesize,indention=0.3cm,margin=0cm]{caption}  

\usepackage{setspace,color}   
\newcounter{nummer}

\newcommand{\ddfr}{\ensuremath{d_{\textup{dF}}}}
\newcommand{\dfr}{\ensuremath{d_{\textup{F}}}}

\newcommand{\Fr}{Fr\'echet }

\newcommand{\SETH}{\ensuremath{\mathsf{SETH}}}

\newcommand{\FS}{\ensuremath{\mathcal{D}}}

\newcommand{\R}{\mathbb{R}}

\newcommand{\N}{\mathbb{N}}
\newcommand{\Oh}{\mathcal{O}}

\newcommand{\pmin}{\tilde{p}}

\newcommand{\IGNORE}[1]{}

\def\min{\operatorname{min}}
\def\max{\operatorname{max}}

\newcommand{\figref}[1]{Figure~\ref{fig:#1}}

\newcommand{\lemref}[1]{Lemma~\ref{lem:#1}}

\newcommand{\lemrefss}[3]{Lemmas~\ref{lem:#1}, \ref{lem:#2}, and~\ref{lem:#3}}

\newcommand{\secref}[1]{Section~\ref{sec:#1}}

\newcommand{\secrefs}[2]{Sections~\ref{sec:#1} and~\ref{sec:#2}}

\renewcommand{\epsilon}{\ensuremath{\varepsilon}}
\newcommand{\eps}{\ensuremath{\varepsilon}}

\newcommand{\temporary}[1]{}

\let\oldsqrt\sqrt
\def\hksqrt{\mathpalette\DHLhksqrt}
\def\DHLhksqrt#1#2{\setbox0=\hbox{$#1\oldsqrt{#2\,}$}\dimen0=\ht0
   \advance\dimen0-0.2\ht0
   \setbox2=\hbox{\vrule height\ht0 depth -\dimen0}%
   {\box0\lower0.4pt\box2}}
\renewcommand{\sqrt}{\hksqrt}

\renewcommand{\le}{\leqslant}
\renewcommand{\ge}{\geqslant}
\renewcommand{\bar}{\overline}

\newcommand*\Let[2]{\State #1 $\gets$ #2}

\global\long\def\R{\mathbb{R}}

\global\long\def\N{\mathbb{N}}

\global\long\def\vis{\mathrm{vis}}

\global\long\def\vis{\mathrm{vis}}

\global\long\def\reach{\mathrm{reach}}

\global\long\def\qmin{q_{\mathrm{min}}}

\global\long\def\pmax{p_{\mathrm{max}}}

\global\long\def\pmin{p_{\mathrm{min}}}

\global\long\def\pcand{p_{\mathrm{cand}}}

\global\long\def\stop{\textsc{stop}}

\global\long\def\pstop{p_{\mathrm{stop}}}

\global\long\def\qstop{q_{\mathrm{stop}}}

\global\long\def\ddF{d_{\mathrm{dF}}}

\global\long\def\bigO{\mathcal{O}}

\global\long\def\MaxGreedyStep{\textsc{MaxGreedyStep}}

\global\long\def\MinGreedyStep{\textsc{MinGreedyStep}}

\global\long\def\GreedyStep{\textsc{GreedyStep}}

\global\long\def\FindSigmaExits{\textsc{Find-\ensuremath{\sigma}-Exits}}

\global\long\def\FindPiExits{\textsc{Find-\ensuremath{\pi}-Exits}}

\global\long\def\PiExitsFromPi{\mbox{\textsc{\ensuremath{\pi}-exits-from-\ensuremath{\pi}}}}

\global\long\def\SigmaExitsFromPi{\mbox{\textsc{\ensuremath{\sigma}-exits-from-\ensuremath{\pi}}}}

\newcommand{\Ind}[1]{\mathrm{D}^{#1}}

\catcode`@=11
\def\nphantom{\v@true\h@true\nph@nt}
\def\nvphantom{\v@true\h@false\nph@nt}
\def\nhphantom{\v@false\h@true\nph@nt}
\def\nph@nt{\ifmmode\def\next{\mathpalette\nmathph@nt}%
  \else\let\next\nmakeph@nt\fi\next}
\def\nmakeph@nt#1{\setbox\z@\hbox{#1}\nfinph@nt}
\def\nmathph@nt#1#2{\setbox\z@\hbox{$\m@th#1{#2}$}\nfinph@nt}
\def\nfinph@nt{\setbox\tw@\null
  \ifv@ \ht\tw@\ht\z@ \dp\tw@\dp\z@\fi
  \ifh@ \wd\tw@-\wd\z@\fi \box\tw@}
\newcount\minute \newcount\hour \newcount\hourMins
\def\now{\minute=\time \hour=\time \divide \hour by 60 \hourMins=\hour \multiply\hourMins by 60
  \advance\minute by -\hourMins \zeroPadTwo{\the\hour}:\zeroPadTwo{\the\minute}}

\def\today{\the\year-\zeroPadTwo{\the\month}-\zeroPadTwo{\the\day}}
\def\zeroPadTwo#1{\ifnum #1<10 0\fi #1}

\title{
Improved approximation for \Fr distance on $c$-packed curves matching conditional lower bounds
}

\author{Karl Bringmann\thanks{Max Planck Institute for Informatics, Campus E1 4, 66123 Saarbr\"ucken, Germany; \texttt{kbringma@mpi-inf.mpg.de}. Karl Bringmann is a recipient of the \emph{Google Europe Fellowship in Randomized Algorithms}, and this research is supported in part by this Google Fellowship.}
\and Marvin K\"unnemann\thanks{Max Planck Institute for Informatics, Campus E1 4, 66123 Saarbr\"ucken, Germany; \texttt{marvin@mpi-inf.mpg.de}. Saarbrücken Graduate School of Computer Science, Germany}
}

\date{}

\begin{document}

\maketitle

\medskip

\begin{abstract}

The \Fr distance is a well-studied and very popular measure of similarity of two curves. The best known algorithms have quadratic time complexity, which has recently been shown to be optimal assuming the Strong Exponential Time Hypothesis (SETH) [Bringmann FOCS'14]. 

To overcome the worst-case quadratic time barrier, restricted classes of curves have been studied that attempt to capture realistic input curves. The most popular such class are $c$-packed curves, for which the \Fr distance has a $(1+\eps)$-approximation in time $\tilde \Oh(c n /\eps)$ [Driemel et al. DCG'12]. In dimension $d \ge 5$ this cannot be improved to $\Oh((cn/\sqrt{\eps})^{1-\delta})$ for any $\delta > 0$ unless SETH fails [Bringmann FOCS'14]. 

In this paper, exploiting properties that prevent stronger lower bounds, we present an improved algorithm with runtime $\tilde \Oh(cn/\sqrt{\eps})$. This is optimal in high dimensions apart from lower order factors unless SETH fails.
Our main new ingredients are as follows: For filling the classical free-space diagram we project short subcurves onto a line, which yields one-dimensional separated curves with roughly the same pairwise distances between vertices. Then we tackle this special case in near-linear time by carefully extending a greedy algorithm for the \Fr distance of one-dimensional separated curves.

\end{abstract}

\newpage

\section{Introduction} \label{sec:intro}

The \Fr distance is a very popular measure of similarity of two given curves and has two classic variants. Roughly speaking, the \emph{continuous \Fr distance} of two curves $\pi,\sigma$ is the minimal length of a leash required to connect a dog to its owner, as they walk without backtracking along $\pi$ and $\sigma$, respectively. In the \emph{discrete \Fr distance} we replace the dog and its owner by two frogs -- in each time step each frog can jump to the next vertex along its curve or stay where it is. 

In a seminal paper in 1991, Alt and Godau introduced the continuous \Fr distance to computational geometry~\cite{AltG95,Godau91}. For polygonal curves $\pi$ and $\sigma$ with $n$ and $m$ vertices\footnote{We always assume that $m \le n$.}, respectively, they presented an $\Oh(nm \log nm)$ algorithm. The discrete \Fr distance was defined by Eiter and Mannila~\cite{EiterM94}, who presented an $\Oh(n m)$ algorithm.

Since then, \Fr distance has become a rich field of research: The literature contains generalizations to surfaces (see, e.g.,~\cite{AltB10}), approximation algorithms for realistic input curves (\cite{AronovHPKW06,AltKW04,DriemelHPW12}), the geodesic and homotopic \Fr distance (see, e.g.,~\cite{ChambersETAL10,Wenk2010geodesic}), and many more variants (see, e.g.,~\cite{BuchinBW09,DriemelHP13,MaheshwariSSZ11,Indyk02}). As a natural measure for curve similarity~\cite{Alt09}, the \Fr distance has found applications in various areas such as signature verification (see, e.g.,~\cite{MunichP99}), map-matching tracking data (see, e.g.,~\cite{BrakatsoulasPSW05}), and moving objects analysis (see, e.g.,~\cite{BuchinBGLL11}).

Apart from log-factor improvements~\cite{AgarwalBAKS13,BuchinBMM14} the quadratic complexity of the classic algorithms for the continuous and discrete \Fr distance are still the state of the art. In fact, the first author recently showed a conditional lower bound: Assuming the Strong Exponential Time Hypothesis (\SETH) there is no algorithm for the (continuous or discrete) \Fr distance in time $\Oh((n m)^{1-\delta})$ for any $\delta > 0$, so apart from lower order terms of the form $n^{o(1)}$ the classic algorithms are optimal~\cite{Bringmann14}.

In attempts to obtain faster algorithms for realistic inputs, various restricted classes of curves have been considered, such as backbone curves~\cite{AronovHPKW06}, $\kappa$-bounded and $\kappa$-straight curves~\cite{AltKW04}, and $\phi$-low density curves~\cite{DriemelHPW12}. The most popular model of realistic inputs are \emph{$c$-packed curves}.  A curve~$\pi$ is $c$-packed if for any point $z \in \R^d$ and any radius $r>0$ the total length of $\pi$ inside the ball $B(z,r)$ is at most $c r$, where $B(z,r)$ is the ball of radius $r$ around~$z$. 
This model has been used for several generalizations of the \Fr distance, such as map matching~\cite{ChenDGNW11}, the mean curve problem~\cite{HarPeledR11}, a variant of the \Fr distance allowing shortcuts~\cite{DriemelHP13}, and \Fr matching queries in trees~\cite{GudmundssonS13}.
Driemel et al.~\cite{DriemelHPW12} introduced $c$-packed curves and presented a $(1+\eps)$-approximation for the continuous \Fr distance in time $\Oh(cn/\eps + cn \log n)$, which works in any $\R^d$, $d \ge 2$.
Assuming \SETH, the following lower bounds have been shown for $c$-packed curves: 
(1) For sufficiently small constant $\eps > 0$ there is no $(1+\eps)$-approximation in time $\Oh((c n)^{1-\delta})$ for any $\delta > 0$~\cite{Bringmann14}. Thus, for constant $\eps$ the algorithm by Driemel et al.\ is optimal apart from lower order terms of the form $n^{o(1)}$. (2) In any dimension $d \ge 5$ and for varying $\eps>0$ there is no $(1+\eps)$-approximation in time $\Oh((c n /\sqrt{\eps})^{1-\delta})$ for any $\delta > 0$~\cite{Bringmann14}. Note that this does not match the runtime of the algorithm by Driemel et al.\ for any $\eps = n^{-b}$ and constant $b >0$. 

In this paper we improve upon the algorithm by Driemel et al.~\cite{DriemelHPW12} by presenting an algorithm that matches the conditional lower bound of~\cite{Bringmann14}.

\begin{theorem}
  For any $0<\eps\le 1$ we can compute a $(1+\eps)$-approximation on $c$-packed curves for the continuous and discrete \Fr distance in time $\tilde \Oh(c n /\sqrt{\eps})$.
\end{theorem}

Specifically, our runtime is $\Oh(\tfrac{cn}{\sqrt{\eps}} \log(1/\eps) + cn \log n)$ for the discrete variant and $\Oh(\tfrac{cn}{\sqrt{\eps}} \log^2(1/\eps) + cn \log n)$ for the continuous variant. 

We want to highlight that in general dimensions (specifically, $d \ge 5$) this runtime is \emph{optimal} (apart from lower order terms of the form $n^{o(1)}$ unless \SETH\ fails~\cite{Bringmann14}). Moreover, we obtained our new algorithm by investigating why the conditional lower bound~\cite{Bringmann14} cannot be improved and exploiting the discovered properties. Thus, the above theorem is the outcome of a synergetic effect of algorithms and lower bounds\footnote{This yields one more reason why conditional lower bounds such as~\cite{Bringmann14} should be studied, as they can show tractable cases and suggest properties that make these cases tractable.}. 

We remark that the same algorithm also yields improved runtime guarantees for other models of realistic input curves, like $\kappa$-bounded and $\kappa$-straight curves, where we are also able to essentially replace $\eps$ by $\sqrt{\eps}$ in the runtime bound. In contrast to $c$-packed curves, it is not clear how far these bounds are from being optimal. See \secref{kappa} for details.

\paragraph{Outline}
We give an improved algorithm that approximately decides whether the \Fr distance of two given curves $\pi,\sigma$ is at most $\delta$. Using a construction of~\cite{DriemelHP13} to search over possible values of~$\delta$, this yields an improved approximation algorithm. We partition our curves into subcurves, each of which is either a \emph{long segment}, i.e., a single segment of length at least $\Lambda = \Theta(\sqrt{\eps} \delta)$, or a \emph{piece}, i.e., a subcurve staying in the ball of radius $\Lambda$ around its initial vertex. Now we run the usual algorithm that explores the reachable free-space (see \secref{preliminaries} for definitions), however, we treat regions spanned by a piece $\pi'$ of $\pi$ and a piece $\sigma'$ of $\sigma$ in a special way. Typically, if $\pi', \sigma'$ consist of $n',m'$ segments then their free-space would be resolved in time $\Oh(n' m')$. Our overall speedup comes from reducing this runtime to $\tilde \Oh(n' + m')$, which is our first main contribution. To this end, we consider the line through the initial vertices of the pieces $\pi',\sigma'$, and project $\pi',\sigma'$ onto this line to obtain curves $\hat \pi, \hat \sigma$. Since $\pi',\sigma'$ are \emph{pieces}, i.e., they stay within distance $\Lambda = \Theta(\sqrt{\eps} \delta)$ of their initial vertices, this projection does not change distances from $\pi$ to $\sigma$ significantly (it follows from the Pythagorean theorem that any distance of approximately $\delta$ is changed, by the projection, by less than $\eps \delta$). Thus, we can replace $\pi',\sigma'$ by $\hat\pi,\hat\sigma$ without introducing too much error. Note that $\hat\pi, \hat\sigma$ are \emph{one-dimensional} curves; without loss of generality we can assume that they lie on $\R$. Moreover, we show how to ensure that $\hat\pi, \hat\sigma$ are \emph{separated}, i.e., all vertices of $\hat\pi$ lie above 0 and all vertices of $\hat\sigma$ lie below 0. Hence, we reduced our problem to resolving the free-space region of one-dimensional separated curves. 

It is known\footnote{We thank Wolfgang Mulzer for pointing us to this result by Matias
Korman and Sergio Cabello (personal communication). To the best of our knowledge this result is not published.
} that the \Fr distance of one-dimensional separated curves can be computed in near-linear time, essentially since we can walk along $\pi$ and $\sigma$ with \emph{greedy steps} to either find a feasible traversal or bottleneck subcurves. However, we face the additional difficulty that we have to resolve the \emph{free-space region} of one-dimensional separated curves, i.e., given entry points on $\hat \pi$ and $\hat \sigma$, compute all exits on $\hat \pi$ and $\hat \sigma$. Our second main contribution is that we present an extension of the known result to handle this much more complex problem.

\paragraph{Organization} We start with basic definitions and techniques borrowed from~\cite{DriemelHP13} in \secref{preliminaries}. In \secref{decider} we present our approximate decision procedure which reduces the problem to one-dimensional separated curves. We solve the latter in \secref{onedim}. In the whole paper, we focus on the continuous \Fr distance. It is straightforward to obtain a similar algorithm for the discrete variant, in fact, then \secref{reductioncontdisc} becomes obsolete, which is why we save a factor of $\log 1/\eps$ in the running time.

\section{Preliminaries}
\label{sec:preliminaries}

For $z \in R^d$, $r>0$ we let $B(z,r)$ be the ball of radius $r$ around $z$. 
For $i,j \in \N$, $i \le j$, we let $[i..j] := \{i,i+1,\ldots,j\}$, which is not to be confused with the real interval $[i,j] = \{x \in \R \mid i \le x \le j\}$. Throughout the paper we fix the dimension $d \ge 2$. A (polygonal) curve~$\pi$ is defined by its vertices $(\pi_1,\ldots,\pi_n)$ with $\pi_p \in \R^d$, $p \in [1..n]$. We let $|\pi| = n$ be the number of vertices of $\pi$ and $\|\pi\|$ be its total length $\sum_{i=1}^{n-1} \|p_i - p_{i+1}\|$. 
We write $\pi_{p..b}$ for the subcurve $(\pi_p,\pi_{p+1},\ldots,\pi_b)$. 
Similarly, for an interval $I = [p..b]$ we write $\pi_I = \pi_{p..b}$.
We can also view $\pi$ as a continuous function $\pi\colon [1,n] \to \R^d$ with $\pi_{p+\lambda} = (1-\lambda) \pi_p + \lambda \pi_{p+1}$ for $p \in [1..n-1]$ and $\lambda \in [0,1]$.
For the second curve $\sigma = (\sigma_1,\ldots,\sigma_m)$ we will use indices of the form $\sigma_{q..d}$ for the reader's convenience.

\paragraph{Variants of the \Fr distance}
Let $\Phi_n$ be the set of all continuous and non-decreasing functions $\phi$ from $[0,1]$ onto $[1,n]$. The \emph{continuous \Fr distance} between two curves $\pi,\sigma$ with $n$ and $m$ vertices, respectively, is defined as
$$ \dfr(\pi,\sigma) := \inf_{\substack{\phi_1 \in \Phi_n \\\phi_2 \in \Phi_m}} \max_{t \in [0,1]} \|\pi_{\phi_1(t)} - \sigma_{\phi_2(t)}\|, $$
where $\|.\|$ denotes the Euclidean distance. We call $\phi := (\phi_1,\phi_2)$ a (continuous) \emph{traversal} of $(\pi,\sigma)$, and say that it has \emph{width} $\max_{t \in [0,1]} \|\pi_{\phi_1(t)} - \sigma_{\phi_2(t)}\|$.

In the discrete case, we let $\Delta_n$ be the set of all non-decreasing functions $\phi$ from $[0,1]$ onto~$[1..n]$. We obtain the \emph{discrete \Fr distance} $\ddfr(\pi,\sigma)$ by replacing $\Phi_n$ and $\Phi_m$ by $\Delta_n$ and $\Delta_m$.
We obtain an analogous notion of a (discrete) \emph{traversal} and its \emph{width}. Note that any $\phi \in \Delta_n$ is a staircase function attaining all values in $[1..n]$. Hence, $(\phi_1(t),\phi_2(t))$ changes only at finitely many points in time $t$. At any such \emph{time step}, we jump to the next vertex in $\pi$ or $\sigma$ or both.

\paragraph{Free-space diagram} The discrete \emph{free-space} of curves $\pi,\sigma$ is defined as $\FS^d_{\le \delta}(\pi,\sigma) := \{(p,q) \in [1..n]\times[1..m] \mid \| \pi_p - \sigma_q\| \le \delta\}$. Note that any discrete traversal of $\pi,\sigma$ of width at most $\delta$ corresponds to a monotone sequence of points in the free-space where at each point in time we increase $p$ or $q$ or both. Because of this property, the free-space is a standard concept used in many algorithms for the \Fr distance. 

The continuous free-space is defined as $\FS_{\le \delta}(\pi,\sigma) := \{(p,q) \in [1,n]\times[1,m] \mid \| \pi_p - \sigma_q\| \le \delta\}$. Again, a monotone path from $(1,1)$ to $(n,m)$ in $\FS_{\le \delta}(\pi,\sigma)$  corresponds to a traversal of width at most $\delta$. 
It is well-known~\cite{AltG95,Godau91} that each \emph{free-space cell} $C_{i,j} := \{(p,q) \in [i,i+1]\times[j,j+1] \mid \| \pi_p - \sigma_q\| \le \delta\}$ (for $i \in [1..n-1], j \in [1..m-1]$) is convex, specifically it is the intersection of an ellipsoid with $[i,i+1]\times[j,j+1]$. In particular, the intersection of the free-space with any interval $[i,i+1] \times \{j\}$ (or $\{i\}\times[j,j+1]$) is an interval $I_{i,j}^h$ (or $I_{i,j}^v$), and for any such interval the subset that is reachable by a monotone path from $(1,1)$ is an interval $R_{i,j}^h$ (or $R_{i,j}^v$). Moreover, in constant time one can solve the following \emph{free-space cell problem}: Given intervals $R^h_{i,j} \subseteq [i,i+1]\times\{j\}, R^v_{i,j} \subseteq \{i\}\times[j,j+1]$, determine the intervals $R^h_{i,j+1} \subseteq [i,i+1]\times\{j+1\}, R^v_{i+1,j} \subseteq \{i+1\}\times[j,j+1]$ consisting of all points that are reachable from a point in $R^h_{i,j} \cup R^v_{i,j}$ by a monotone path within the free-space cell~$C_{i,j}$.
Solving this problem for all cells from lower left to upper right we determine whether $(n,m)$ is reachable from $(1,1)$ by a monotone path and thus decide whether the \Fr distance is at most~$\delta$.

\paragraph{From approximate deciders to approximation algorithms}

An \emph{approximate decider} is an algorithm that, given curves $\pi,\sigma$ and $\delta>0,0<\eps\le 1$, returns one of the outputs (1) $\dfr(\pi,\sigma) > \delta$ or (2) $\dfr(\pi,\sigma) \le (1+\eps) \delta$. In any case, the returned answer has to be correct.
In particular, if $\delta < \dfr(\pi,\sigma) \le (1+\eps)\delta$ the algorithm may return either of the two outputs. 

Let $D(\pi,\sigma,\delta,\eps)$ be the runtime of an approximate decider and set $D(\pi,\sigma,\eps) := \max_{\delta > 0} D(\pi,\sigma,\delta,\eps)$. We assume polynomial dependence on $\eps$, in particular, that there are constants $0 < c_1 < c_2 < 1$ such that for any $1<\eps \le 1$ we have $c_1 D(\pi,\sigma,\eps/2) \le D(\pi,\sigma,\eps) \le c_2 D(\pi,\sigma,\eps/2)$.
Driemel et al.~\cite{DriemelHPW12} gave a construction of a $(1+\eps)$-approximation for the \Fr distance given an approximate decider. (This follows from~\cite[Theorem 3.15]{DriemelHPW12} after replacing their concrete approximate decider with runtime ``$\Oh(N(\eps,\pi,\sigma))$'' by any approximate decider with runtime $D(\pi,\sigma,\eps)$.)

\begin{lemma} \label{lem:approxfromdecider}
  Given an approximate decider with runtime $D(\pi,\sigma,\eps)$ we can construct a $(1+\eps)$-approximation for the \Fr distance with runtime 
  $\Oh\big(D(\pi,\sigma,\eps) + D(\pi,\sigma,1) \log n\big)$.
\end{lemma}

\section{The approximate decider}
\label{sec:decider}

\begin{figure}
\centering
\subfloat[This figure illustrates our partitioning of a curve into \emph{pieces} (contained in dashed circles) and \emph{long segments} (bold edges).]{
  \includegraphics[width=0.45\textwidth]{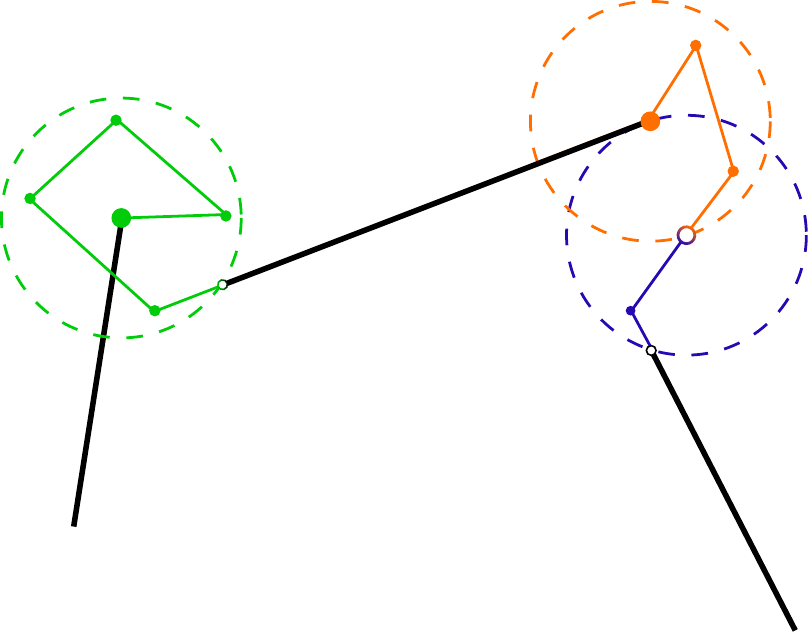}
  \label{fig:pieces}
}
\quad
\subfloat[The free-space problem for pieces $\pi'$ and $\sigma'$ in the free-space diagram of $\pi$ and $\sigma$. Given entry intervals on the lower and left boundary of the region, compute exit intervals on the upper and right boundary.]{
  \includegraphics[width=0.45\textwidth]{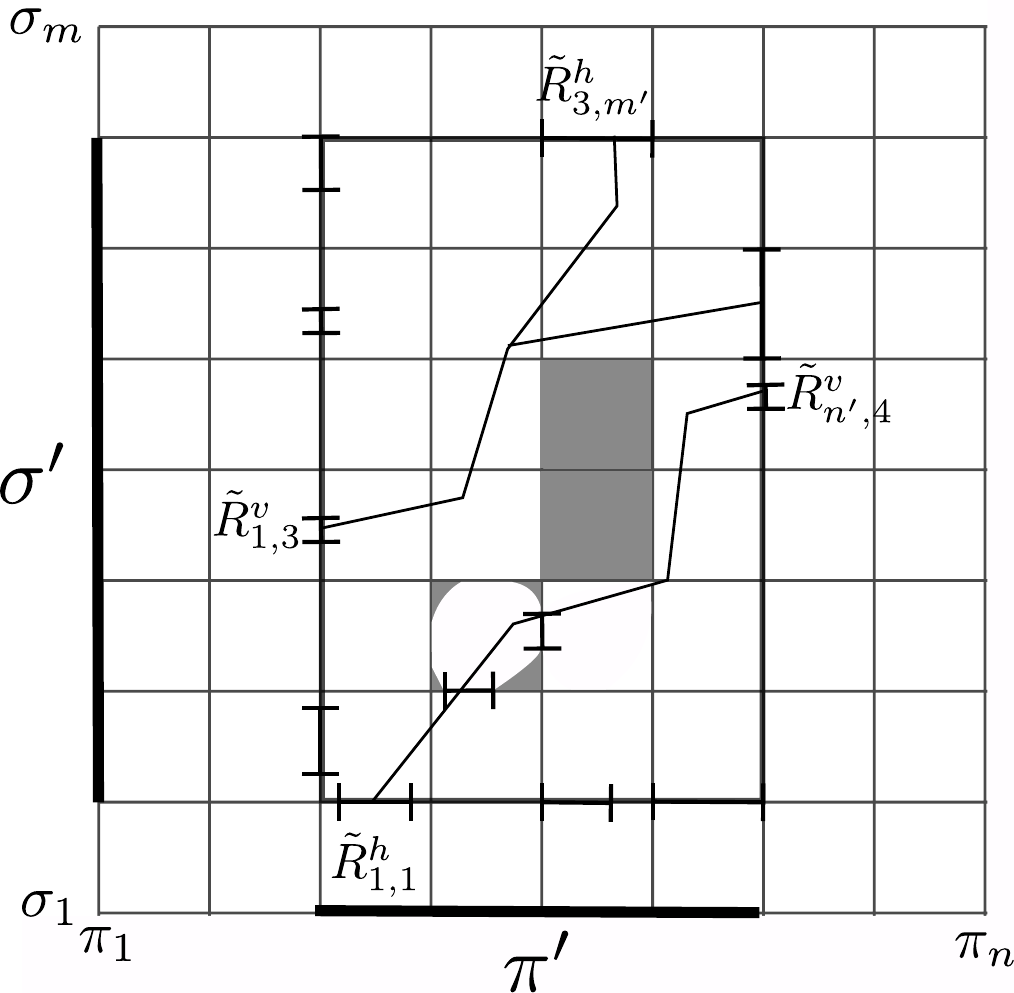}
  \label{fig:freespace}
}
\caption{Definition and treatment of pieces.}
\end{figure}

Let $\pi,\sigma$ be curves for which we want to (approximately) decide whether $\dfr(\pi,\sigma) > \delta$ or $\dfr(\pi,\sigma) \le (1+\eps)\delta$. 
We modify the curve $\pi$ by introducing new vertices as follows. Start with the initial vertex $\pi_1$ as current vertex. If the segment following the current vertex has length at least $\Lambda = \Lambda_{\eps,\delta} := \min\{\tfrac 12 \sqrt{\eps}, \tfrac 14\}\cdot \delta$ then mark this segment as \emph{long} and set the next vertex as the current vertex. Otherwise follow $\pi$ from the current vertex $\pi_x$ to the first point $\pi_y$ such that $\|\pi_{x} - \pi_y\| = \Lambda$ (or until we reach the last vertex of $\pi$). If $\pi_y$ is not a vertex, but lies on some segment of~$\pi$, then introduce a new vertex at~$\pi_y$. Mark $\pi_{x..y}$ as a \emph{piece} of $\pi$ and set $\pi_y$ as current vertex. Repeat until $\pi$ is completely traversed. Since this procedure introduces at most $|\pi|$ new vertices and does not change the shape of $\pi$, with slight abuse of notation we call the resulting curve again $\pi$ and set $n := |\pi|$. This partitions $\pi$ into subcurves $\pi^1,\ldots,\pi^k$, with $\pi^s = \pi_{p_s..b_s}$, where every part $\pi^s$ is either (see also \figref{pieces})
\begin{itemize}
  \item a \emph{long segment}: $b_s = p_s+1$ and $\|\pi_{p_s}-\pi_{b_s}\| \ge \Lambda$, or 
  \item a \emph{piece}: $\|\pi_{p_s} - \pi_{b_s}\| = \Lambda$ and $\|\pi_{p_s} - \pi_{x}\| < \Lambda$ for all $x \in [p_s,b_s)$.
\end{itemize}
Note that the last piece actually might have distance $\|\pi_{p_s} - \pi_{b_s}\|$ less than $\Lambda$, however, for simplicity we assume equality for all pieces (in fact, a special handling of the last piece would only be necessary in \lemref{cpackedcomplex}).
Similarly, we introduce new vertices on $\sigma$ and partition it into subcurves $\sigma^1,\ldots,\sigma^\ell$, with $\sigma^t = \sigma_{q_t..d_t}$, each of which is a long segment or a piece. Let $m := |\sigma|$.

We do not want to resolve each free-space cell on its own, as in the standard decision algorithm for the \Fr distance. Instead, for any pair of pieces we want to consider the free-space region spanned by the two pieces at once, see \figref{freespace}. This is made formal by the following subproblem.

\begin{problem}[Free-space region problem]
\label{prob:freespace}
Given $\delta>0$, $0<\eps\le 1$, curves $\pi, \sigma$ with $n$ and $m$ vertices, and \emph{entry intervals} $\tilde R^h_{i,1} \subseteq [i,i+1] \times \{1\}$ for $i \in [1..n)$ and $\tilde R^v_{1,j} \subseteq \{1\} \times [j,j+1]$ for $j \in [1..m)$, compute \emph{exit intervals} $\tilde R^h_{i,m} \subseteq [i,i+1]\times\{m\}$ for $i \in [1..n)$ and $\tilde R^v_{n,j} \subseteq \{n\} \times [j,j+1]$ for $j \in [1..m)$ such that (1) the exit intervals contain all points reachable from the entry intervals by a monotone path in $\FS_{\le \delta}(\pi,\sigma)$ and (2) all points in the exit intervals are reachable from the entry intervals by a monotone path in $\FS_{\le (1+\eps)\delta}(\pi,\sigma)$. 
\end{problem}

To stress that we work with approximations, we denote reachable intervals by $\tilde R$ instead of $R$ in the remainder of the paper.

The standard solution to the free-space region problem would split it up into $n \cdot m$ free-space cells and resolve each cell in constant time, resulting in an $\Oh(n \cdot m)$ algorithm (this solves the problem even exactly, i.e., for $\eps=0$). Restricted to pieces, we will show the following improvement, which will yield the desired overall speedup of a factor of $\sqrt{\eps}$.
\begin{lemma} \label{lem:fspp}
  If $\pi$ and $\sigma$ are pieces then the free-space region problem can be solved in time $\Oh((n + m) \log^2 1/\eps)$.
\end{lemma}
We will prove this lemma in \secrefs{fspp}{onedim}.

\begin{algenv}\label{alg:freespace} Using an algorithm for the free-space region problem on pieces as in \lemref{fspp}, we obtain an approximate decider for the \Fr distance a follows. We create a directed graph which has a node $v_{s,t}$ for every region $[p_s,b_s]\times[q_t,d_t]$ spanned by pieces $\pi^s$ and $\sigma^t$, and a node $u_{i,j}$ for every remaining region $[i,i+1]\times[j,j+1]$ (which is not contained in any region spanned by two pieces), $i \in [1..n)$, $j \in [1..m)$. We add edges between two nodes whenever their regions touch (i.e., have a common interval~$I$ on their boundary), and direct this edge from the region that is to the left or below $I$ to the other one. With each node $u_{i,j}$ we store the entry intervals $\tilde R^h_{i,j}$ and $\tilde R^v_{i,j}$, and with each node $v_{s,t}$ we store the entry intervals $\tilde R^h_{i,q_t} \subseteq [i,i+1] \times \{q_t\}$ for $i \in [p_s..b_s)$ and $\tilde R^v_{p_s,j} \subseteq \{p_s\} \times [j,j+1]$ for $j \in [q_t..d_t)$.
After correctly initializing the outer reachability intervals $\tilde R^h_{i,1}$ and $\tilde R^v_{1,j}$,
we follow any topological ordering of this graph. For any node $u_{i,j}$, we resolve its region by solving the corresponding free-space \emph{cell} problem in constant time. For any node $v_{s,t}$, we solve the corresponding free-space \emph{region} problem on $\pi' = \pi^s, \sigma' = \sigma^t$ (and $\delta' = \delta, \eps' = \eps$) using \lemref{fspp}. Finally, we return $\dfr(\pi,\sigma) \le (1+\eps)\delta$ if $(n,m) \in \tilde R^h_{n-1,m}$ and $\dfr(\pi,\sigma) > \delta$ otherwise. 
\end{algenv}

\begin{lemma}
  Algorithm~\ref{alg:freespace} is a correct approximate decider.
\end{lemma}
\begin{proof} 
  Observe that if $(n,m) \in \tilde R^h_{n-1,m}$ then there exists a monotone path from $(1,1)$ to $(n,m)$ in $\FS_{\le (1+\eps)\delta}(\pi,\sigma)$, which implies $\dfr(\pi,\sigma) \le (1+\eps)\delta$. If $\dfr(\pi,\sigma) \le \delta$ then there is a monotone path from $(1,1)$ to $(n,m)$ in $\FS_{\le \delta}(\pi,\sigma)$, implying $(n,m) \in \tilde R^h_{n-1,m}$.
\end{proof}

In the above algorithm we can ignore \emph{unreachable} nodes, i.e., nodes where all stored entry intervals would be empty. To this end, we fix a topological ordering by mapping a node corresponding to a region $[x_1,x_2]\times[y_1,y_2]$ to $x_2+y_2$ and sorting by this value ascendingly. This yields $n+m$ layers of nodes, where the order within each layer is arbitrary. For each layer we build a dictionary data structure (a hash table), in which we store only the reachable nodes of this layer. This allows to quickly enumerate all reachable nodes of a layer. The total overhead for managing the $n+m$ dictionaries is $\Oh(n+m)$.

Let us analyze the runtime of the obtained approximate decider. Let $S$ be the set of non-empty free-space cells $C_{i,j}$ of $\FS_{\le (1+\eps) \delta}(\pi,\sigma)$ such that $i$ or $j$ is not contained in a piece. Moreover, let $T$ be the set of all pairs $(s,t)$ such that $\pi^s,\sigma^t$ are pieces with initial vertices within distance $(1+\eps)\delta + 2 \Lambda$. Define $N(\pi,\sigma,\delta,\eps) := |S| + \sum_{(s,t) \in T} (|\pi^s| + |\sigma^t|)$ and set $N(\pi,\sigma,\eps) := \max_{\delta > 0} N(\pi,\sigma,\delta,\eps)$. Since the algorithm considers only \emph{reachable} cells and any reachable cell is also non-empty, the cost over all free-space cell problems solved by our approximate decider is bounded by $\Oh(|S|)$. Since every reachable (thus non-empty) region spanned by two pieces has initial points within distance $(1+\eps)\delta + 2\Lambda$, the second term bounds the cost over all free-space region problems on pieces (apart from the $\log^2 1/\eps$ factor). Hence, we obtain the following.

\begin{lemma} \label{lem:runtime}
  The approximate decider has runtime $D(\pi,\sigma,\eps) = \Oh(N(\pi,\sigma,\eps) \cdot \log^2 1/\eps)$.
\end{lemma}

\subsection{The free-space complexity of $c$-packed curves}

Recall that a curve~$\pi$ is $c$-packed if for any point $z \in \R^d$ and any radius $r>0$ the total length of $\pi$ inside the ball $B(z,r)$ is at most $c r$. 

\begin{lemma} \label{lem:cpackedcomplex}
  Let $\pi,\sigma$ be $c$-packed curves with $n$ vertices in total and $\eps>0$. Then $N(\pi,\sigma,\eps) = \Oh(cn/\sqrt{\eps})$.
\end{lemma}
\begin{proof}
  Our proof uses a similar argument as~\cite[Lemma 4.4]{DriemelHP13}.
  Let $\delta > 0$ be arbitrary. First consider the set $S$ of non-empty free-space cells $C_{i,j}$ of $\FS_{\le (1+\eps) \delta}(\pi,\sigma)$ such that $i$ or $j$ is not contained in a piece. Then one of the segments $\pi_{i..i+1}$ and $\sigma_{j..j+1}$ is long, i.e., of length at least $\Lambda = \min\{\tfrac 12 \sqrt{\eps}, \tfrac 14\}\cdot \delta$. We charge the cell $C_{i,j}$ to the shorter of the two segments. Let us analyze how often any segment $v = \pi_{i..i+1}$ can be charged. Consider the ball $B$ of radius $r := \tfrac 12 \|v\| + (1+\eps)\delta + \max\{\|v\|, \Lambda\}$ centered at the midpoint of~$v$. Every segment $u = \sigma_{j..j+1}$ with $(i,j) \in S$, which charges $v$, is of length at least $\mu := \max\{\|v\|, \Lambda\}$ (since it is longer than $v$ and a long segment) and contributes at least $\mu$ to the total length of $\sigma$ in $B$. Since $\sigma$ is $c$-packed, the number of such charges is at most 
  $$ \frac{\|\sigma \cap B\|}{\mu} \le \frac {cr}{\mu} \le \frac{c(\tfrac 12 \|v\| + (1+\eps)\delta + \max\{\|v\|, \Lambda\})}{\max\{\|v\|, \Lambda\}} \le \tfrac 32 c + \frac{c(1+\eps)\delta}{\min\{\tfrac 12 \sqrt{\eps}, \tfrac 14\}\cdot \delta} = \Oh\Big(\frac{c}{\sqrt{\eps}}\Big).$$
  Thus, the contribution of $|S|$ to the free-space complexity $N(\pi,\sigma,\eps)$ is $\Oh(cn/\sqrt{\eps})$.
  
  Let $T$ be the set of all pairs $(s,t)$ such that $\pi^s,\sigma^t$ are pieces of $\pi,\sigma$ with initial vertices within distance $(1+\eps)\delta + 2\Lambda$, and consider $\Sigma := \sum_{(s,t) \in T} (|\pi^s| + |\sigma^t|)$. We distribute $\Sigma$ over the segments of $\pi,\sigma$ by charging 1 to every segment of $\pi^s$ and $\sigma^t$ for any pair $(s,t) \in T$. Let us analyze how often any segment $v$ of a piece $\pi^s$ can be charged. Consider the ball $B'$ of radius $r' := (1+\eps)\delta + 3\Lambda$ around the initital vertex $\pi_{p_s}$ of $\pi^s$. Since $\|\sigma^t\| \ge \Lambda$, for any $(s,t) \in T$ the piece $\sigma^t$ contributes at least $\Lambda$ to the total length of $\sigma$ in $B'$. Since $\sigma$ is $c$-packed, the number of such charges to $v$ is at most
  $$ \frac{\|\sigma \cap B'\|}{\Lambda} \le \frac{cr'}{\Lambda} = \frac{c(1+\eps+\tfrac 32 \sqrt{\eps})}{\min\{\tfrac 12 \sqrt{\eps}, \tfrac 14\}} = \Oh\Big( \frac{c}{\sqrt{\eps}} \Big).$$
  Hence, the contribution of $\Sigma$ to the free-space complexity $N(\pi,\sigma,\eps)$ is also at most $\Oh(cn/\sqrt{\eps})$, which finishes the proof.
\end{proof}

Combining \lemrefss{cpackedcomplex}{runtime}{approxfromdecider}, we obtain an approximation algorithm for the \Fr distance with running time $\Oh(\tfrac{cn}{\sqrt{\eps}} \log^2 1/\eps + cn \log n) = \tilde \Oh(\tfrac{cn}{\sqrt{\eps}})$, as desired.

\subsection{The free-space complexity of $\kappa$-bounded and $\kappa$-straight curves}
\label{sec:kappa}

\begin{definition}
  Let $\kappa \ge 1$ be a given parameter. A curve $\pi$ is \emph{$\kappa$-straight} if for any $p,b \in [1,|\pi|]$ we have $\|\pi_{p..b}\| \le \kappa \|\pi_p - \pi_b\|$. A curve $\pi$ is \emph{$\kappa$-bounded} if for all $p,b$ the subcurve $\pi_{p..b}$ is contained in $B(\pi_p,r) \cup B(\pi_b,r)$, where $r = \tfrac \kappa 2 \|\pi_p- \pi_b\|$. 
\end{definition}

The following lemma from~\cite{DriemelHP13} allows us to transfer our speedup for $c$-packed curves directly to $\kappa$-straight curves.

\begin{lemma}
  A $\kappa$-straight curve is $2\kappa$-packed.
\end{lemma}

In the remainder of this section we consider $\kappa$-bounded curves, closely following~\cite[Sect.\ 4.2]{DriemelHP13}.

\begin{lemma} \label{lem:kappaboundedhelp}
  Let $\delta >0$, $0<\eps\le 1$, $\lambda > 0$, and let $\pi$ be a $\kappa$-bounded curve with disjoint subcurves $\pi^1,\ldots,\pi^k$, where $\pi^s = \pi_{p_s..b_s}$ and $\|\pi_{p_s} - \pi_{b_s}\| \ge \lambda$ for all $s$. Then for any $z \in \R^d$, $r>0$ the number of subcurves~$\pi^s$ intersecting $B(z,r)$ is bounded by $\Oh(\kappa^d(1+r/\lambda)^d)$.
\end{lemma}
\begin{proof}
  Let $\pi^{s_1},\ldots,\pi^{s_\ell}$ be the subcurves that intersect the ball $B = B(z,r)$. Let $X = \{s_1,s_3,\ldots,\}$ be the odd indices among the intersecting subcurves. For all $s \in X$ pick any point $\pi_{x_s}$ in $\pi^s \cap B$. Between any points $\pi_{x_s}, \pi_{x_{s'}}$ there must lie an even subcurve $\pi^{s_{2i}}$. As the endpoints of this even subcurve have distance at least $\lambda$, we have $\|\pi_{x_s} - \pi_{x_{s'}}\| \ge \lambda/(\kappa+1)$. Otherwise the even part would not fit into $B(\pi_{x_s},r)\cup B(\pi_{x_{s'}},r)$ which has diameter $(\kappa+1) \|\pi_{x_s} - \pi_{x_{s'}}\|$.
  Hence, the balls $B(\pi_{x_s},\lambda/2(\kappa+1))$ are disjoint and contained in $B(z,r+\lambda)$. A standard packing argument now shows that $\ell \le 2\cdot (r+\lambda)^d / (\lambda/2(\kappa+1))^d = \Oh(\kappa^d (1+r/\lambda)^d)$.
\end{proof}

\begin{lemma}
  For any $\kappa$-bounded curves $\pi,\sigma$ with $n$ vertices in total, $0<\eps\le 1$, we have $N(\pi,\sigma,\eps) = \Oh((\kappa/\sqrt{\eps})^d n)$.
\end{lemma}
\begin{proof}
  Let $\delta > 0$ and consider the partitionings into long segments and pieces $\pi^1,\ldots,\pi^k$, $\sigma^1,\ldots,\sigma^\ell$ computed by our algorithm. Then $\sigma^t = \sigma_{q_t..d_t}$ satisfies $\|\sigma_{q_t} - \sigma_{d_t}\| \ge \Lambda = \min\{\tfrac 12 \sqrt{\eps}, \tfrac 14\}\cdot \delta$ for all~$t$. 
  We use the same charging scheme as in \lemref{cpackedcomplex}. Consider any segment $v$ of a piece $\pi^s$. The segment~$v$ can be charged by a part $\sigma^t$ which is either a long segment or a piece. In both cases, $\sigma^t$ intersects the ball $B$ centered at the midpoint of $\|v\|$ with radius $r := (1+\eps)\delta + 2 \Lambda$. By \lemref{kappaboundedhelp} with $\lambda := \Lambda$, the number of such charges is bounded by $\Oh((\kappa/\sqrt{\eps})^d)$. 
  
  Now consider any long segment $v$ of $\pi$. The segment $v$ can be charged by segments of $\sigma$ which are longer than $v$. Any such charging gives rise to a long segment $\sigma^t$ intersecting the ball $B$ centered at the midpoint of $v$ of radius $r := (1+\eps)\delta + \tfrac 12 \|v\|$. By \lemref{kappaboundedhelp} with $\lambda := \|v\|$, the number of such charges is bounded by $\Oh(\kappa^d (\tfrac32 + (1+\eps)\delta/\|v\|)^d) = \Oh((\kappa/\sqrt{\eps})^d)$, since $\|v\| \ge \Lambda = \min\{\tfrac 12 \sqrt{\eps}, \tfrac 14\}\cdot \delta$. 
  
  Hence, every segment of $\pi$ is charged $\Oh((\kappa/\sqrt{\eps})^d)$ times; a symmetric statement holds for~$\sigma$.
\end{proof}

Plugging the above lemma into \lemref{approxfromdecider} we obtain the following result. The best previously known runtime was $\Oh((\kappa/\eps)^d n + \kappa^d n \log n)$~\cite{DriemelHP13}.

\begin{theorem}
  For any $0<\eps\le 1$ there is a $(1+\eps)$-approximation for the continuous and discrete \Fr distance on $\kappa$-bounded curves with $n$ vertices in total in time $\Oh((\kappa/\sqrt{\eps})^d n \log^2 1/\eps + \kappa^d n \log n) = \tilde\Oh((\kappa/\sqrt{\eps})^d n)$.
\end{theorem}

\subsection{Solving the free-space region problem on pieces} 
\label{sec:fspp}

It remains to prove \lemref{fspp}. 
Let $(\pi,\sigma,\delta,\eps)$ be an instance of the free-space region problem, where $n := |\pi|$, $m := |\sigma|$, with $\|\pi_1 - \pi_x\|, \|\sigma_1 - \sigma_y\| \le \Lambda_{\eps,\delta} = \Lambda$ for any $x \in[1,n]$, $y \in [1,m]$ (and entry intervals $\tilde R^h_{i,1} \subseteq [i,i+1] \times \{1\}$ for $i \in [1..n)$ and $\tilde R^v_{1,j} \subseteq \{1\} \times [j,j+1]$ for $j \in [1..m)$).
We reduce this instance to the free-space region problem on \emph{one-dimensional separated} curves, i.e., curves $\hat \pi, \hat \sigma$ in $\R$ such that all vertices of $\hat \pi$ lie above 0 and all vertices of $\hat \sigma$ lie below~0.

Since $\pi$ and $\sigma$ stay within distance $\Lambda$ of their initial vertices, if their initial vertices are within distance $\|\pi_1 - \sigma_1\| \le \delta - 2\Lambda$ then all pairs of points in $\pi,\sigma$ are within distance $\delta$. In this case, we find a translation of $\pi$ making $\|\pi_1 - \sigma_1\| = \delta - 2\Lambda$ and all pairwise distances are still at most~$\delta$. This ensures that the curves $\pi,\sigma$ are contained in disjoint balls of radius $\Lambda \le \tfrac 14 \delta$ centered at their initial vertices.

\begin{figure}
\centering
\includegraphics[width=0.7\textwidth]{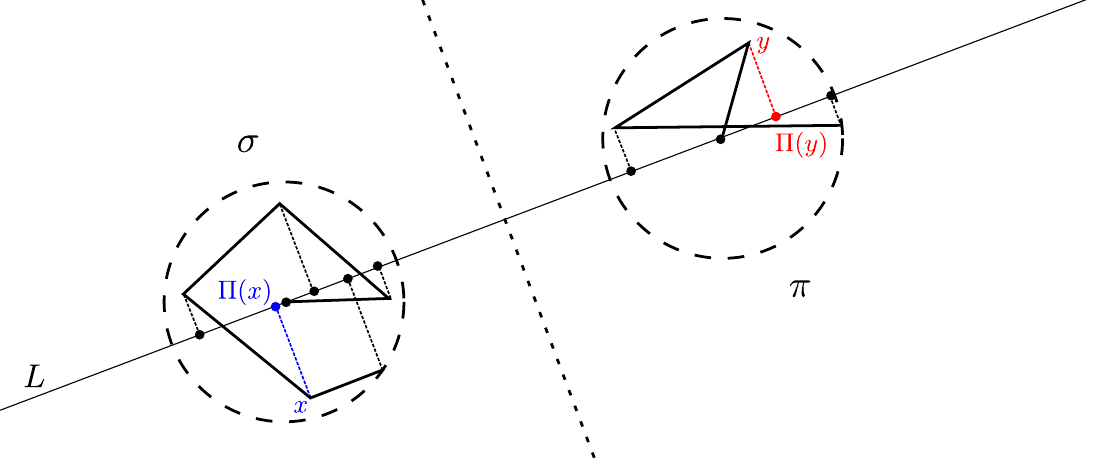}
\caption{Projection of the pieces $\pi,\sigma$ onto the line $L$ through their initial vertices. This yields one-dimensional separated curves $\hat \pi, \hat \sigma$.}
\label{fig:projection}
\end{figure}

Consider the line $L$ through the initial vertices $\pi_1$ and $\sigma_1$. Denote by $\Pi\colon \R^d \to L$ the projection onto~$L$. Now, instead of the pieces $\pi,\sigma$ we consider their projections $\hat\pi := \Pi(\pi) = (\Pi(\pi_{1}),\ldots,\Pi(\pi_{n}))$ and $\hat \sigma := \Pi(\sigma) =  (\Pi(\sigma_{1}),\ldots,\Pi(\sigma_{m}))$, see \figref{projection}. Note that after rotation and translation we can assume that $\hat \pi$ and $\hat \sigma$ lie on $\R \subset \R^d$ and $\hat \pi$ and $\hat \sigma$ are separated by $0 \in \R$ (since $\pi$ and $\sigma$ are contained in disjoined balls centered on $L$). 
Now we solve the free-space region problem on $\hat \pi$, $\hat \sigma$, $\hat \delta := \delta$, and $\hat \eps := \tfrac 12 \eps$ (with the same entry intervals $\tilde R^h_{i,j}, \tilde R^v_{i,j})$).

\begin{lemma}
  Any solution to the the free-space region problem on $(\hat \pi, \hat \sigma, \hat \delta, \hat \eps)$ solves the free-space region problem on $(\pi, \sigma, \delta, \eps)$.
\end{lemma}
\begin{proof}
  Let $x,y$ be vertices of $\pi,\sigma$, respectively. Clearly, $\|\Pi(x) - \Pi(y) \| \le \| x - y \|$. Hence, any monotone path in $\FS_{\le \delta}(\pi,\sigma)$ yields a monotone path in $\FS_{\le \delta}(\hat \pi, \hat \sigma) = \FS_{\le \hat \delta}(\hat \pi, \hat \sigma)$, so it will be found.
  
  Note that $x$ and $y$ have distance at most $\Lambda$ to $L$. Since $\Pi(x) - \Pi(y)$ and $x-\Pi(x) - (y- \Pi(y))$ are orthogonal, we can use the Pythagorean theorem to obtain 
  $$\|x - y\| = \sqrt{\|\Pi(x) - \Pi(y)\|^2 + \|x-\Pi(x) - (y- \Pi(y))\|^2} \le \sqrt{\|\Pi(x) - \Pi(y)\|^2 + (2\Lambda)^2}.$$ 
  Hence, any monotone path in $\FS_{\le (1+\hat\eps)\hat \delta}(\hat \pi, \hat \sigma)$ yields a monotone path in $\FS_{\le \alpha}(\pi,\sigma)$ with $\alpha \le \sqrt{(1+\hat\eps)^2 \hat \delta^2 + (2\Lambda)^2}$. Plugging in $\hat \delta = \delta$, $\hat \eps = \tfrac 12 \eps$, and $\Lambda = \min\{\tfrac 12 \sqrt{\eps}, \tfrac 14\}\cdot \delta$ we obtain $\alpha \le \sqrt{(1+\tfrac 12 \eps)^2 + \eps} \cdot \delta \le (1+\eps)\, \delta$. Thus, the desired guarantees for the free-space region problem are satisfied.
\end{proof}

We will show the following lemma in \secref{onedim}, concluding the proof of \lemref{fspp}.

\begin{lemma} \label{lem:osfsp}
  The free-space region problem on one-dimensional separated curves 
  can be solved in time $\Oh((n + m) \log^2 1/\eps)$.
\end{lemma}

\section{On one-dimensional separated curves} \label{sec:onedim}

In this section, we show how to solve the free-space region problem on one-dimensional separated curves in time $\Oh((n+m) \log^2 1/\eps)$, i.e., we prove \lemref{osfsp}. 

First, in \secref{reductioncontdisc}, we show how to reduce this problem to a \emph{discrete} version, meaning that we can eliminate the continuous \Fr distance and only consider the much simpler discrete \Fr distance (for general curves such a reduction is not known to exist, but we only need it for one-dimensional separated curves). Moreover, we simplify our curves further by rounding the vertices.
This yields a reduction to the following subproblem. Note that we no longer ask for an approximation algorithm.

\begin{problem}[Reduced free-space problem]
 Given one-dimensional separated curves $\pi,\sigma$ with $n,m$ vertices and all vertices being multiples of $\tfrac 13 \eps \delta$, and given an entry set $E \subseteq [1..n]$, compute the exit set $F^\pi \subseteq [1..n]$ consisting of all points $f$ such that $\ddfr(\pi_{e..f},\sigma) \le \delta$ for some $e \in E$ and the exit set $F^\sigma \subseteq [1..m]$ consisting of all points $f$ such that $\ddfr(\pi_{e..n},\sigma_{1..f}) \le \delta$ for some $e \in E$.
\end{problem}

\begin{lemma} \label{lem:redfsp}
  The reduced free-space problem can be solved in time $\Oh((n+m) \log 1/\eps)$.
\end{lemma}
As a second step, we prove the above lemma. We first consider the special case of $E = \{1\}$ and the problem of deciding whether $n \in F^\pi$, i.e., the lower left corner $(1,1)$ of the free-space is the only entry point and we want to determine whether the upper right corner $(n,m)$ is an exit. This is equivalent to deciding whether the discrete \Fr distance of $\pi,\sigma$ is at most $\delta$, which is known to have a near-linear time algorithm as $\pi,\sigma$ are one-dimensional and separated (see the footnote in the introduction for details). We present a greedy algorithm for this special case in \secref{1d-greedydec}. To extend this to the reduced free-space problem, we prove useful structural properties of one-dimensional separated curves in \secref{1d-composition}. With these, we first solve the problem of determining the exit set $F^\pi$ assuming $E = \{1\}$ in \secref{1d-single}. Then we show for general $E \subseteq [1..n]$ how to compute $F^\pi$ (\secref{1d-multiPi}) and~$F^\sigma$ (\secref{1d-multiSigma}).

\subsection{Reduction from the continuous to the discrete case}
\label{sec:reductioncontdisc}

Essentially we use the following lemma to reduce the \emph{continuous} free-space region problem on one-dimensional separated curves to the \emph{discrete} reduced free-space problem.

\begin{lemma} \label{lem:conttodisc}
  Let $\pi,\sigma$ be one-dimensional separated curves with subcurves $\pi_{p..b}, \sigma_{q..d}$. Then we have $\dfr(\pi_{p..b},\sigma_{q..d}) = \ddfr(\pi_{p..b},\sigma_{q..d})$. In particular, assume that we subdivide any segments of $\pi,\sigma$ by adding new vertices, which yields new curves $\pi',\sigma'$ with subcurves $\pi'_{p'..b'},\sigma'_{q'..d'}$ that are subdivisions of $\pi_{p..b},\sigma_{q..d}$. Then we have $\ddfr(\pi'_{p'..b'},\sigma'_{q'..d'}) = \ddfr(\pi_{p..b},\sigma_{q..d}) = \dfr(\pi_{p..b},\sigma_{q..d})$.
\end{lemma}
\begin{proof}
  It is known that $\dfr(\pi,\sigma) \le \ddfr(\pi,\sigma)$ holds for all curves $\pi,\sigma$. Thus, we only need to show that any continuous traversal $\phi = (\phi_1,\phi_2)$ of $\pi_{p..b},\sigma_{q..d}$ can be transformed into a discrete traversal with the same width. We adapt $\phi$ as follows. For any point in time $t \in [0,1]$, if $\phi_1(t)$ is at a vertex of~$\pi$ we set $\phi'_1(t) := \phi_1(t)$. Otherwise $\phi_1(t)$ is in the interior of a segment $\pi_{i..i+1}$ of $\pi$. Let $j \in \{i,i+1\}$ minimize $\pi_j$. We set $\phi_1'(t) := j$. Observe that $\phi_1'$ indeed is a non-decreasing function from $[0,1]$ onto $[1..n]$.
  A similar construction, where we round to the value $j \in \{i,i+1\}$ maximizing~$\sigma_j$, yields $\phi_2'$ and we obtain a discrete traversal $\phi' = (\phi'_1,\phi'_2)$. The width of $\phi'$ is at most the width of $\phi$ since we rounded in the right way, i.e., we have $\pi(\phi'_1(t)) \le \pi(\phi_1(t))$ and $\sigma(\phi'_2(t)) \ge \sigma(\phi_2(t))$ so that $\|\pi(\phi'_1(t)) - \sigma(\phi'_2(t))\| \le \|\pi(\phi_1(t)) - \sigma(\phi_2(t))\|$ for all $t \in [0,1]$.
  
  Note that the discrete \Fr distance is in general not preserved under subdivision of segments, but the continuous \Fr distance is.
  Thus, the second statement follows from the first one, $\ddfr(\pi_{p..b},\sigma_{q..d}) = \dfr(\pi_{p..b},\sigma_{q..d}) = \dfr(\pi'_{p'..b'},\sigma'_{q'..d'}) = \ddfr(\pi'_{p'..b'},\sigma'_{q'..d'})$.
\end{proof}

The above lemma allows the following trick. Consider any finite sets $E \subseteq [1,n]$ and $F \subseteq [1,n]$. Add $\pi_x$ as a vertex to $\pi$ for any $x \in E \cup F$, with slight abuse of notation we say that $\pi$ now has vertices at $\pi_i$, $i \in [1..n]$, and $\pi_x$, $x \in E \cup F$. Mark the vertices $\pi_x$, $x \in E$, as entries. Now solve the reduced free-space problem instance $(\pi,\sigma,E)$. This yields the set $F^\pi$ of all values $f \in F$ such that there is an $e \in E$ with $\ddfr(\pi_{e..f},\sigma) \le \delta$, which by \lemref{conttodisc} is equivalent to $\dfr(\pi_{e..f},\sigma) \le \delta$. Thus, we computed all exit points in $F$ given entry points in $E$, with respect to the continuous \Fr distance. This is already near to a solution of the free-space region problem, however, we have to cope with entry and exit \emph{intervals}. 

For the full reduction we need two more arguments. First, we can replace all non-empty input intervals $\tilde R^h_{i,1}$ by the leftmost point $(y_i,1)$ in $\tilde R^h_{i,1} \cap \FS_{\le \delta}(\pi,\sigma)$, specifically, we show that any traversal starting in a point in $\tilde R^h_{i,1}$ can be transformed into a traversal starting in $(y_i,1)$. Thus, we add $\pi_{y_i}$ as a vertex and mark it as an entry to obtain a finite and small set of entry points. Second, for any segment $\pi_{i..i+1}$ we call a point $f \in [i,i+1]$ \emph{reachable} if there is an $e \in E$ with $\dfr(\pi_{e..f},\sigma) \le \delta$.  We show that if $f$ is reachable then essentially all points $f' \in [i,i+1]$ with $\pi_{f'} \le \pi_f$ are also reachable. Thus, the set of reachable points is an interval with one trivial endpoint, and we only need to search for the other endpoint of the interval, which can be done by binary search. Moreover, we can parallelize all these binary searches, as solving one reduced free-space problem can answer for every segment of $\pi$ whether a particular point on this segment is reachable (after adding this point as a vertex).
To make these binary searches finite, we round all vertices of $\pi$ and $\sigma$ to multiples of $\gamma := \tfrac 13 \eps \delta$ and only search for exit points that are multiples of~$\gamma$. This is allowed since the free-space region problem only asks for an approximate answer. A similar procedure yields the exits on $\sigma$ reachable from entries on $\pi$, and determining the exits reachable from entries on $\sigma$ is a symmetric problem. Since for the binary searches we reduce to $\Oh(\log 1/\eps)$ instances of the reduced free-space problem, \lemref{osfsp} follows from \lemref{redfsp}.

In the following we present the details of this approach.
Let $\pi,\sigma$ be one-dimensional separated curves, i.e., they are contained in $\R$, all vertices of $\pi$ lie above 0, and all vertices of $\sigma$ lie below 0. Let $n = |\pi|$, $m = |\sigma|$, $\delta > 0$ and $0< \eps \le 1$. Consider entry intervals $\tilde R^h_{i,1} \subseteq [i,i+1] \times \{1\}$ for $i \in [1..n)$ and $\tilde R^v_{1,j} \subseteq \{1\} \times [j,j+1]$ for $j \in [1..m)$. 
We reduce this instance of the free-space region problem to $\Oh(\log 1/\eps)$ instances of the reduced free-space problem.

First we change $\pi,\sigma$ as follows.
(1) Let $Z \subset \R$ be the set of all integral multiples\footnote{Without loss of generality we assume $1/\eps \in \N$ so that $\delta \in Z$.} of $\gamma := \tfrac 13 \eps \delta$.
We round all vertices of $\pi,\sigma$ to values in~$Z$, where we round down everything in $\pi$ and round up in $\sigma$, yielding curves $\pi',\sigma'$. (2) Let $I \subseteq [1..n)$ be the set of all $i$ with nonempty $\tilde R^h_{i,1} \cap \FS_{\le \delta}(\pi',\sigma')$. For any $i \in I$ let $(y_i,1)$ be the leftmost point in $\tilde R^h_{i,1} \cap \FS_{\le \delta}(\pi',\sigma')$ and note that $\pi'_{y_i}$ is also a multiple of $\gamma$. Add $\pi'_{y_i}$ as a vertex to $\pi'$ and mark it as an entry. With slight abuse of notation, we say that $\pi'$ now has its vertices at $\pi'_i$, $i \in [1..n]$ and $\pi'_{y_i}$, $i \in I$. We let $E = \{y_i \mid i \in I\}$ be the indices of the entry vertices. Note that $(\pi',\sigma',E)$ can be computed in time $\Oh(n+m)$.

For every $i \in [1..n)$ consider the multiples of $\gamma$ on $\pi'_{i..i+1}$, i.e., $S_i := \{x \in [i,i+1] \mid \pi'_x \in Z\}$. 
Note that $S_i$ forms an arithmetic progression, specifically $S_i = \{i,i+1/t_i,i+2/t_i,\ldots,i+1\}$ for some $t_i \in \N$, since $\pi'_i,\pi'_{i+1}$ are in $Z$ and $\pi'_x$ is a linear function in $x$. Thus, $S_i$ and subsequences of $S_i$ can be handled efficiently, we omit these details in the following. 
We want to determine the set $F_i$ of all $f \in S_i$ such that there is an $e \in E$ with $\ddfr(\pi'_{e..f},\sigma') \le \delta$. 
We first argue that $F_i$ is of an easy form.

\begin{lemma} \label{lem:Fi}
  If $F_i$ is non-empty then we have $F_i = [a,b] \cap Z$ for some $a,b \in S_i$ with $\{a,b\} \cap \{i,y_i,i+1\} \ne \emptyset$ (or $\{a,b\} \cap \{i,i+1\} \ne \emptyset$ if $y_i$ does not exist).
\end{lemma}
\begin{proof}
  We show that if any $f \in S_i$ is reachable, i.e., there is an $e \in E$ with $\ddfr(\pi'_{e..f},\sigma') \le \delta$, then any $f' \in S_i$ with $\pi'_{f'} \le \pi'_f$ and $y_i \not\in (f',f]$ is also reachable. This proves the claim. Let $\phi$ be any traversal of $\pi'_{e..f},\sigma'$ of width at most $\delta$. Note that $e \le f'$, since $y_i \not\in (f',f]$ and $y_i$ is the only entry on the segment containing $f$ and $f'$. If $f' \le f$ then we change $\phi$ to stop at $\pi'_{f'}$ once it arrives at this point, and we traverse the remaining part of $\sigma$ staying fixed at $\pi'_{f'}$. Since $\pi'_{f'} \le \pi'_f$ this does not increase the width of the traversal and shows that $f'$ is also reachable. If $f' > f$ then we append a traversal to $\phi$ that stays fixed at $\sigma'_m$ but walks in $\pi'$ from $\pi'_f$ to $\pi'_{f'}$. Again since $\pi'_{f'} \le \pi'_f$ this does not increase the width of the traversal and shows that $f'$ is also reachable.
\end{proof}

Note that by solving the reduced free-space problem on $(\pi',\sigma',E)$ we decide for each $f \in [n] \cup \{y_i \mid i \in I\}$ whether there is an $e \in E$ with $\ddfr(\pi'_{e..f},\sigma')\le \delta$. By the above lemma, this yields one of the endpoints of the interval $F_i$, say $a$, and we only have to determine the other endpoint, say $b$. In the special case $\pi'_i = \pi'_{i+1}$ we even determined both endpoints already, so from now on we can assume $\pi'_i \ne \pi'_{i+1}$ so that $|S_i|< \infty$.
We search for the other endpoint of $F_i$ using a binary search over $S_i$. To test whether any $z \in S_i$  is in $F_i$, we add $\pi'_z$ as a vertex of $\pi'$ and solve the reduced free-space problem on $(\pi',\sigma',E)$. If $z$ is in the output set $F^\pi$ then it is in $F_i$. 

Note that any vertex $\pi'_x > \delta$ on $\pi'$ does not have any point of $\sigma$ within distance $\delta$, which is preserved by setting $\pi'_x := 2\delta$. Thus, we can assume that $\pi'$ takes values in $[0,2\delta]$, which implies $|S_i| \le \Oh(1/\eps)$, so that our binary search needs $\Oh(\log 1/\eps)$ steps. 
Moreover, note that we can parallelize these binary searches, since we can add a vertex $z_i$ on every subcurve $\pi'_{i..i+1}$, so that one call to the reduced free-space problem determines for every $z_i$ whether it is reachable. Here we use \lemref{conttodisc}, since we need that further subdivision of some segments of $\pi'$ does not change the discrete \Fr distance. Note that since we add $\Oh(n)$ vertices to $\pi'$ and since we need $\Oh(\log 1/\eps)$ steps of binary search, \lemref{redfsp} implies a total runtime of $\Oh((n+m) \log^2 1/\eps)$. 

We thus computed $F_i = [a,b] \cap Z$ with $a,b \in S_i$. We extend $F_i$ slightly to $F_i' = [a',b'] \cap Z$ by including the neighboring elements of $a$ and $b$ in $S_i$. Finally, we set $\tilde R^h_{i,m}(\pi) := [a',b'] \times \{m\}$. A similar procedure adding entries $E^{\sigma}$ on $\sigma'$ and doing a binary search over exits on $\pi'$ yields an interval $\tilde R^h_{i,m}(\sigma)$ consisting of points $(f,m) \in [i,i+1] \times \{m\}$ such that there is an $e \in E^\sigma$ with $\ddfr(\pi'_{1..f},\sigma'_{e..m}) \le \delta$. We set $\tilde R^h_{i,m} := \tilde R^h_{i,m}(\pi) \cup \tilde R^h_{i,m}(\sigma)$, which will be again an interval (which follows from the proof of \lemref{Fi}). A symmetric algorithm determines $\tilde R^v_{n,j}$ for $j \in [1..m)$.

We show that we correctly solve the given free-space region problem instance.

\begin{lemma}
  The computed intervals are a valid solution to the given free-space region instance.
\end{lemma}
\begin{proof}
  Let $\phi$ be any monotone path in $\FS_{\le \delta}(\pi,\sigma)$ that starts in a point $(p,1) \in \tilde R^h_{j,1}$ and ends in $(b,m)$, witnessing that $\ddfr(\pi_{p..b},\sigma) \le \delta$. After rounding down $\pi$ to $\pi'$ and rounding up $\sigma$ to $\sigma'$, $\phi$ is still a monotone path in $\FS_{\le \delta}(\pi',\sigma')$. Moreover, we can prepend a path from $(y_j,1)$ to $(p,1)$ to $\phi$, since $\tilde R^h_{j,1} \cap \FS_{\le \delta}(\pi',\sigma')$ is an interval containing $(y_j,1)$ and $(p,1)$. 
  Let $r$ be the value of $\pi_b$ rounded down to a multiple of $\gamma$. This value $r$ is attained at some point $\pi_{f}$ on the same segment $\pi_{i..i+1}$ as $\pi_b$. If $f \le b$ then we change $\phi$ to stop at $\pi_f$ whenever it reaches this point. 
  If $f > b$ then we change $\phi$ by appending a path from $(b,m)$ to $(f,m)$. In any case, this yields a monotone path in $\FS_{\le \delta}(\pi',\sigma')$ from $(y_j,1)$ to $(f,m)$. Since such a continuous traversal is equivalent to a discrete traversal by \lemref{conttodisc}, we have $f \in F_i$. 
  By the construction of $F_i'$, the point $(b,m)$ will be contained in the output $\tilde R^h_{i,m}(\pi)$, so we find the reachable exit $(b,m)$ as desired. A similar argument with entries on $\sigma$ shows that we satisfy property (1) of the free-space region problem.
  
  Consider any point $(f,m)$ in the output set $\tilde R^h_{i,m}(\pi)$. By the construction of $F_i'$, there is a point $b$ on the same segment as $f$ with $\|\pi'_b - \pi'_f\| \le \gamma$ and there is an entry $e \in E$ with $\ddfr(\pi'_{e..b},\sigma') \le \delta$, witnessed by a traversal $\phi$. If $b \le f$ we change $\phi$ so that it stops at $\pi'_b$ once it reaches this point. If $b > f$ we change $\phi$ by appending a path from $(b,m)$ to $(f,m)$. In any case, this shows $\ddfr(\pi'_{e..f},\sigma') \le \delta + \gamma$. Since $\pi',\sigma'$ are rounded versions of $\pi,\sigma$ where all vertices are moved by less than $\gamma$, we obtain $\ddfr(\pi_{e..f},\sigma) \le \delta + 3\gamma = (1+\eps)\delta$. Thus, any point $(f,m)$ in the output set is reachable form the entry sets by a monotone path in $\FS_{\le (1+\eps)\delta}(\pi,\sigma)$, which together with a similar argument for entries on $\sigma$ proves that we satisfy property (2) of the free-space problem.
\end{proof}

\subsection{Greedy Decider for the Fréchet Distance of One-Dimensional Separated
Curves}
\label{sec:1d-greedydec}

In the remainder of the paper all indices of curves will be integral.
Let $\pi=(\pi_{1},\dots,\pi_{n})$ and $\sigma=(\sigma_{1},\dots,\sigma_{m})$
be two separated polygonal curves in $\mathbb{R}$, i.e., $\pi_{i}\ge0\ge\sigma_{j}$.
For indices $1\le i\le n$ and $1\le j\le m$, define $\vis_{\sigma}(i,j):=\{k\mid k\ge j\text{ and }\sigma_{k}\ge\pi_{i}-\delta\}$
as the index set of vertices on $\sigma$ that are later in sequence
than $\sigma_{j}$ and are still in distance $\delta$ to $\pi_{i}$ (i.e,
\emph{seen} by $\pi_{i}$) and, likewise, $\vis_{\pi}(i,j):=\{k\mid k\ge i\text{ and }\pi_{k}\le\sigma_{j}+\delta\}$.
Hence, the set of points that we may reach on $\sigma$ by starting
in $(\pi_{i},\sigma_{j})$ and staying in $\pi_{i}$ can be defined
as the longest contiguous subsequence $[j+1 .. j+k]$ such that
$[j+1 .. j+k]\subseteq\vis_{\sigma}(i,j)$. Let $\reach_{\sigma}(i,j):=[j+1 .. j+k]$
denote this subsequence and let $\reach_{\pi}(i,j)$ be defined symmetrically.
Note that $\pi_{i}\le\pi_{i'}$ implies that $\vis_{\sigma}(i,j)\supseteq\vis_{\sigma}(i',j)$,
however the converse does not necessarily hold. Also, $\vis_{\sigma}(i,j)\nsupseteq\vis_{\sigma}(i',j)$
implies that $\vis_{\sigma}(i,j)\subsetneq\vis_{\sigma}(i',j)$ and
$\pi_{i}>\pi_{i'}$.

The visibility sets established above enable us to define a greedy
algorithm for the Fréchet distance of $\pi$ and $\sigma$.
Let $1\le p\le n$ and $1\le q\le m$ be arbitrary indices on $\sigma$
and $\pi$. We say that $p'$ is a \emph{greedy step on $\pi$ from
$(p,q)$}, written $p'\gets\GreedyStep_{\pi}(\pi_{p..n},\sigma_{q..m})$,
if $p'\in\reach_{\pi}(p,q)$ and $\vis_{\sigma}(i,q)\subseteq\vis_{\sigma}(p',q)$
holds for all $p\le i\le p'$. Symmetrically, $q'\in\reach_{\sigma}(p,q)$
is a \emph{greedy step on $\sigma$ from $(p,q)$}, if $\vis_{\pi}(p,i)\subseteq\vis_{\pi}(p,q')$
for all $q\le i\le q'$. In pseudo code, $\GreedyStep_{\pi}(\pi_{p.. n},\sigma_{q.. m})$
denotes a function that returns an arbitrary greedy step $p'$ on
$\pi$ from $(p,q)$ if such an index exists and returns an error
otherwise (symmetrically for $\sigma$). See \figref{greedystep}.

Consider the following greedy algorithm:

\begin{algorithm}[H]
\begin{algorithmic}[1]
\State $p\gets 1,q\gets 1$
\Repeat
\If{$p'\gets\GreedyStep_\pi(\pi_{p.. n},\sigma_{q.. m})$} \Let{$p$}{$p'$}
\EndIf
\If{$q'\gets\GreedyStep_\sigma(\pi_{p.. n},\sigma_{q..m})$} \Let{$q$}{$q'$}
\EndIf
\Until{no greedy step was found in the last iteration}
\If{$p=n$ and $q=m$} \Return{$\ddF(\pi,\sigma)\le \delta$}
\Else{  \Return{$\ddF(\pi,\sigma)>\delta$}}
\EndIf
\end{algorithmic}

\caption{Greedy algorithm for the Fréchet distance of separated curves $\pi_{1..n}$
and $\sigma_{1.. m}$ in $\mathbb{R}$}

\label{alg:greedy}
\end{algorithm}

\begin{figure}
\centering
\includegraphics[width=0.65\textwidth]{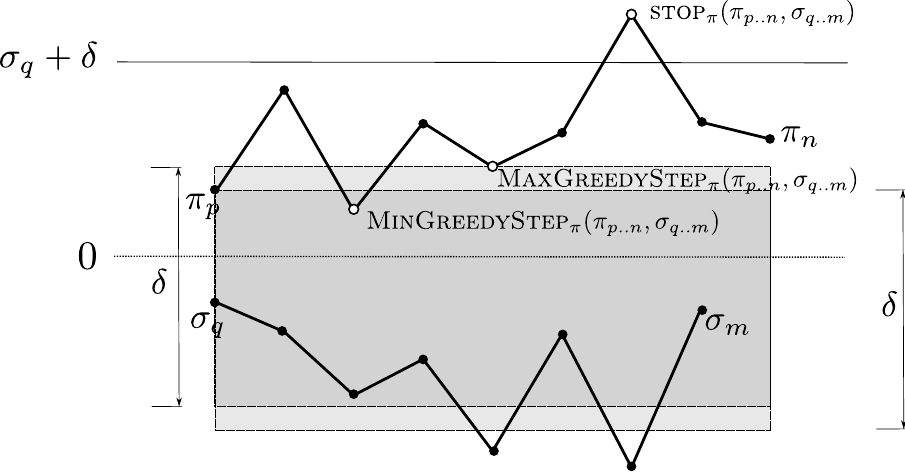}
\caption{An illustration of greedy steps. For better visibility, the one-dimensional separated curves $\pi,\sigma$ are drawn in the plane by mapping $\pi_i$ to $(i,\pi_i)$. 
In particular, the results of $\MinGreedyStep_\pi(\pi_{p..n},\sigma_{q..m})$, $\MaxGreedyStep_\pi(\pi_{p..n},\sigma_{q..m})$, and $\stop_{\pi}(\pi_{p..n},\sigma_{q..m})$ are shown.}
\label{fig:greedystep}
\end{figure}

\begin{theorem}
\label{thm:greedy}Let $\pi$ and $\sigma$ be separated curves in
$\mathbb{R}$ and $\delta > 0$. Algorithm~\ref{alg:greedy} decides whether $\ddF(\pi,\sigma)\le\delta$
in time $\bigO((n+m)\log(nm))$.
\end{theorem}
We will first prove the correctness of the algorithm in Lemma~\ref{lem:GreedyCorrectness}
below and postpone the discussion how to implement the algorithm efficiently
to Section~\ref{sub:implementGreedy}.

\subsubsection{Correctness}

Note that Algorithm~\ref{alg:greedy} considers potentially only
very few points of the curve explicitly during its execution. Call
the indices $(p,q)$ of point pairs considered in some iteration of
the algorithm (for any choice of greedy steps, if more than one exists)
\emph{greedy (point) pairs} and all points contained in some such pair
\emph{greedy points (of $\pi$ and $\sigma$).} The following useful
monotonicity property holds: If some greedy point\emph{ }on $\pi$
sees a point on $\sigma$ that is yet to be traversed, all following
greedy points on $\pi$ will see it \emph{until it is traversed}.
\begin{lemma}
\label{lem:monotonicity}Let $(p_{1},q_{1}),\dots,(p_{i},q_{i})$
be the greedy point pairs considered in the iterations $1,\dots,i$.
It holds that
\begin{enumerate}
\item $\vis_{\sigma}(\ell,q_{i})\subseteq\vis_{\sigma}(p_{i},q_{i})$ for
all $1\le\ell\le p_{i}$, and
\item $\vis_{\pi}(p_{i},\ell)\subseteq\vis_{\pi}(p_{i},q_{i})$ for all
$1\le\ell\le q_{i}$.
\end{enumerate}
\end{lemma}
\begin{proof}
Let $k<i$. We first show that $\vis_{\sigma}(\ell,q_{i})\subseteq\vis_{\sigma}(p_{k+1},q_{i})$
holds for all $p_{k}\le\ell<p_{k+1}$. If $p_{k}=p_{k+1}$, the claim
is immediate. Otherwise $p_{k+1}$ is the result of a greedy step
on $\pi$. By definition of visibility, we have $\vis_{\sigma}(\ell,q_{i})=\vis_{\sigma}(\ell,q_{k})\cap[q_{i}..m]\subseteq\vis_{\sigma}(p_{k+1},q_{k})\cap[q_{i}.. m]=\vis_{\sigma}(p_{k+1},q_{i})$,
where the inequality follows from $p_{k+1}$ being a greedy step from
$(p_{k},q_{k})$.

For arbitrary $\ell\le i$, let $k<i$ be such that $p_{k}\le\ell<p_{k+1}$.
Then $\vis_{\sigma}(\ell,q_{i})\subseteq\vis_{\sigma}(p_{k+1},q_{i})\subseteq\vis_{\sigma}(p_{k+2},q_{i})\subseteq\cdots\subseteq\vis_{\sigma}(p_{i},q_{i})$.
The second statement is symmetric.
\end{proof}

We will exploit this monotonicity to prove that if Algorithm~\ref{alg:greedy}
finds a greedy point pair that allows no further greedy steps, then no feasible
traversal of $\pi$ and $\sigma$ exists. We derive an even stronger
statement using the following notion: For a greedy point pair $(p,q)$,
define $\stop_{\pi}(\pi_{p.. n},\sigma_{q.. m}):=\max(\reach_{\pi}(p,q)\cup\{p\})+1$
as the index of the first point after $\pi_{p}$ on $\pi$ which is
not seen by $\sigma_{q}$, or $n+1$ if no such index exists.  Let
$\stop_{\sigma}$ be defined symmetrically.
\begin{lem}
[Correctness of Algorithm~\ref{alg:greedy}]\label{lem:GreedyCorrectness}Let
$(p,q)$ be a greedy point of $\pi$ and $\sigma$, $\pstop:=\stop_{\pi}(\pi_{p..n},\sigma_{q.. m})$
and $\qstop:=\stop_{\sigma}(\pi_{p.. n},\sigma_{q.. m})$. If
on both curves, no greedy step from $(p,q)$ exists, then $\ddF(\pi,\sigma)>\delta$.

In particular, if $\qstop<m$, then for all $1\le p'\le n$, we have
that $\ddF(\pi_{1.. p'},\sigma_{1..\qstop})>\delta$ and if $\pstop<n$,
then $\ddF(\pi_{1..\pstop},\sigma_{1.. q'})>\delta$ for all $1\le q'\le m$.
\end{lem}
Note that the correctness of Algorithm~\ref{alg:greedy} follows
immediately: If the algorithm is stuck, then $\ddF(\pi,\sigma)>\delta$.
Otherwise, it finds a feasible traversal.
\begin{proof}
[Proof of Lemma~\ref{lem:GreedyCorrectness}]Consider the case that
no greedy step from $(p,q)$ exists, then the following \emph{stuckness
conditions} have to hold:

\begin{enumerate}
\item For all $p'\in\reach_{\pi}(p,q)$, we have $\vis_{\sigma}(p',q)\subsetneq\vis_{\sigma}(p,q)$,
and
\item for all $q'\in\reach_{\sigma}(p,q)$, we have $\vis_{\pi}(p,q')\subsetneq\vis_{\pi}(p,q)$.
\end{enumerate}

In this case, we can extend the monotonicity property of \lemref{monotonicity} to include all reachable
and the first unreachable point.
\begin{claim}
\label{claim:stuckness}
If the stuckness
conditions hold for $(p,q)$, then we have $\vis_{\sigma}(i,q)\subseteq\vis_{\sigma}(p,q)$
for all $1\le i\le\pstop$. In particular, if $\pi_{p}$ does not
see $\sigma_{\ell}$ for some $\ell>q$, then no vertex $\pi_{i}$
with $1\le i\le\pstop$ sees $\sigma_{\ell}.$ The symmetric statement
holds for $\sigma$.
\end{claim}

\begin{proof}
By the monotonicity of the previous claim, $\vis_{\sigma}(i,q)\subseteq\vis_{\sigma}(p,q)$
holds for all $i\le p$. The first of the stuckness conditions implies $\vis_{\sigma}(i,q)\subseteq\vis_{\sigma}(p,q)$ for
all $p<i<\pstop$. If $\pstop=n+1$, this already completes the proof
of the claim. Otherwise, note that $\pi_{\pstop}>\pi_{p}$, since
otherwise $\pstop\in\reach_{\pi}(p,q)$. Hence $\vis_{\sigma}(\pstop,q)\subseteq\vis_{\sigma}(p,q)$
holds as well.
\end{proof}

We distinguish the following cases that may occur under the stuckness
conditions:

\emph{Case 1: }$\pstop\le n$ or $\qstop\le m$. Without loss of generality,
let $\pstop \le n$ (the other case is symmetric).  Assume for contradiction
that a feasible traversal $\phi$ of $\pi_{1..\pstop}$ and $\sigma_{1.. q'}$
exists for some $1\le q'\le m$. In $\phi$, at some point in time
we have to move in $\pi$ from $\pstop-1$ to $\pstop$ while moving
in $\sigma_{1..q'}$ from some $\sigma_{\ell'}$ to $\sigma_{\ell}$
where $\ell'\in\{\ell-1,\ell\}$ and $\sigma_{\ell}$ sees $\pi_{\pstop}$.
Since $\sigma_{q}$ does not see $\pi_{\pstop}$, the previous claim
shows that $\ell>q_{\mathrm{stop}}$. If $\qstop=m+1$ or $\qstop<q'$,
this is impossible, yielding a contradiction. Otherwise, to do this
transition, in some earlier step we have to move in $\sigma$ from
$q_{\mathrm{stop}}-1$ to $q_{\mathrm{stop}}$ while moving in $\pi$
from $\pi_{k'}$ to $\pi_{k}$ for some $k<\pstop$ and $k'\in\{k-1,k\}$.
However, by definition $q_{\mathrm{stop}}\notin\vis_{\sigma}(p,q)$,
hence Claim~\ref{claim:stuckness} implies that the transition is
illegal, since $\pi_{k}$ does not see $\sigma_{\qstop}$. This is
a contradiction. By a symmetric argument, it holds that $\ddF(\pi_{1..p'},\sigma_{1..\qstop})>\delta$.

\emph{Case 2:} $\pstop=n+1$ and $\qstop=m+1$. In this case, $\reach_{\pi}(p,q)=[p+1.. n]$
and $\reach_{\sigma}(p,q)=[q+1.. n]$. By stuckness conditions,
there exist an index $\pmax>p$ such that no $\sigma_{q'}$ with $q'>q$
sees $\pi_{\pmax}$ and an index $\qmin$ such that no $\pi_{p'}$
with $p'>p$ sees $\sigma_{\qmin}$. Assume for contradiction that
a feasible traversal $\phi$ exists. In $\phi$, at some point in
time $t$, we have to cross either (1) from $\pi_{p}$ to $\pi_{p+1}$
while moving in $\sigma$ from $\sigma_{\ell'}$ to $\sigma_{\ell}$
with $\ell\le q+1\le\qmin$ and $\ell'\in\{\ell-1,\ell\}$ or (2)
from $\sigma_{q}$ to $\sigma_{q+1}$ while moving from $\pi_{\ell'}$
to $\pi_{\ell}$ with $\ell\le p+1\le\pmax$ and $\ell'\in\{\ell-1,\ell\}$.
In the first case, $\ell<\qmin$ holds, since $\pi_{p+1}$ does not
see $\sigma_{\qmin}$. For all consecutive times $t'\ge t$, $\phi$
is in a point $\pi_{p'}$ ($p'\ge p+1$) that does not see $\sigma_{\qmin}$,
which still has to be traversed, leading to a contradiction. Symmetrically,
in the second case, for all times $t'\ge t$, $\phi$ is in a point
$\sigma_{q'}$ $(q'\ge q+1)$ that does not see $\pi_{\pmax}$, which
still has to be traversed.

This concludes the proof of Lemma~\ref{lem:GreedyCorrectness}.
\end{proof}

\subsubsection{Implementing greedy steps}

\label{sub:implementGreedy}

To prove Theorem~\ref{thm:greedy}, it remains to show how to implement
the algorithm to run in time $\bigO((n+m)\log(nm))$. We make use
of geometric range search queries. The classic technique of fractional
cascading~\cite{Lueker78,ChazelleG86, willard78} provides a data structure $D$ with the following
properties: (i) Given
$n$ points ${\cal P}$ in the plane, $D({\cal P})$ can be constructed
in time $\bigO(n\log n)$ and (ii) given a query rectangle $Q:=I_{1}\times I_{2}$
with intervals $I_{1}$ and $I_{2}$, find and return $q\in Q\cap{\cal P}$
with minimal $y$-coordinate, or report that no such point exists,
in time $\bigO(\log n)$. Here, each interval $I_{i}$ may be open,
half-open or closed.

By invoking the above data structure on ${\cal P}:=\{(i,\pi_{i})\mid i \in [1 \ldots n]\}$
for a given curve $\pi=\pi_{1..n}$ (as well as all three rotations
of ${\cal P}$ by multiples of $90\degree$), we obtain a datastructure
$\Ind{\pi}$ such that:
\begin{enumerate}
\item $\Ind{\pi}$ can be constructed in time $\bigO(n\log n)$, 
\item the query $\Ind{\pi}.\textsc{minIndex}([x_{1},x_{2}],[p,b])$ ($\Ind{\pi}.\textsc{maxIndex}([x_{1},x_{2}],[p,b])$)
returns the minimum (maximum) index $p\le i\le b$ such that $x_{1}\le\pi_{i}\le x_{2}$
in time $\bigO(\log n)$, and
\item the query $\Ind{\pi}.\textsc{minHeight}([x_{1},x_{2}],[p,b])$ ($\Ind{\pi}.\textsc{maxHeight}([x_{1},x_{2}],[p,b])$)
returns the minimum (maximum) height $x_{1}\le\pi_{i}\le x_{2}$ such
that $p\le i\le b$ in time $\bigO(\log n)$.
\end{enumerate}
The queries extend naturally to open and half-open intervals. If no
index exists in the queried range, all of these operations return
the index $\infty$. We will use the corresponding data structure
$\Ind{\sigma}$ for $\sigma$ as well.

With these tools, we implement the following basic operations for
arbitrary subcurves $\pi':=\pi_{p..b}$ and $\sigma':=\sigma_{q..d}$
of $\pi$ and $\sigma$. See also \figref{greedystep}.
\begin{enumerate}
\item \textbf{Stopping points $\stop_{\pi}(\pi',\sigma')$.} For points $p,q$, $\stop_{\pi}(\pi',\sigma'):=\max(\reach_{\pi'}(p,q)\cup\{p\})+1$
returns the index of the first point after $\pi_{p}$ on $\pi'$ which
is not seen by $\sigma_{q}$, or $b+1$ if no such index exists.
\begin{algorithm}[h]
\begin{algorithmic}[1]
\Function{$\stop_\pi$}{$\pi_{p..b}, \sigma_{q..d}$}
\Let{$\pstop$}{$\Ind{\pi}.\textsc{minIndex}((\sigma_q+\delta,\infty),[p,b]$)} \Comment{First non-visible point on $\pi$}
\If{$\pstop <\infty$} \Return $\pstop$ \Else{ \Return $b+1$} \EndIf
\EndFunction
  \end{algorithmic}

\caption{Finding the stopping point}

\label{alg:stop}
\end{algorithm}

\item \textbf{Minimal greedy steps $\MinGreedyStep_{\pi}(\pi',\sigma')$.
}This function returns the smallest index $p'\in\reach_{\pi'}(p,q)$
such that $\vis_{\sigma'}(p',q)\supseteq\vis_{\sigma'}(p,q)$ or reports
that no such index exists. 
\begin{algorithm}[h]
\begin{algorithmic}[1]
\Function{$\MinGreedyStep_\pi$}{$\pi_{p..b}, \sigma_{q..d}$}
\Let{$\qmin$}{$\Ind{\sigma}.\textsc{minHeight}([\pi_p-\delta,\infty),[q,d])$} \Comment{Lowest still visible point on $\sigma$}
\Let{$\pcand$}{$\Ind{\pi}.\textsc{minIndex}((-\infty,\sigma_{\qmin}+\delta],[p+1,d]$)} \Comment{If $p'$ exists, it is $\pcand$}
\Let{$\pstop$}{$\stop_\pi(\pi_{p..b},\sigma_{q..d})$} \Comment{First non-visible point on $\pi$}
\If{$\pcand < \pstop$} \Return $\pcand$ 
\Else{ \Return ``No greedy step possible.'' \Comment{$\pi_{\pcand}$ not reachable from $\pi_p$ while staying in $\sigma_q$}}
\EndIf
\EndFunction
  \end{algorithmic}

\caption{Minimal greedy step}

\label{alg:smallgreedy}
\end{algorithm}

\item \textbf{Maximal greedy steps $\MaxGreedyStep_{\pi}(\pi',\sigma')$.}
Let $p'\in\reach_{\pi'}(p,q)$ be such that (i) $p'$ is the largest
index maximizing $|\vis_{\sigma'}(z,q)|$ among all $z\in\reach_{\pi'}(p,q)$
and (ii) $\vis_{\sigma'}(p',q)\supseteq\vis_{\sigma'}(p,q)$. If $p'$
exists, $\MaxGreedyStep_{\pi}$ returns this value, otherwise it reports
that no such index exists. Note that if $p'$ exists, then by definition
there is no greedy step on $\pi$ starting from $(p',q)$, i.e., this
step is a maximal greedy step.
\begin{algorithm}[h]
\begin{algorithmic}[1]
\Function{$\MaxGreedyStep_\pi$}{$\pi_{p..b}, \sigma_{q..d}$}
\Let{$\qmin$}{$\Ind{\sigma}.\textsc{minHeight}([\pi_p-\delta,\infty),[q,d])$} \Comment{Lowest still visible point on $\sigma$}
\Let{$\pstop$}{$\stop_\pi(\pi_{p..b},\sigma_{q..d})$} \Comment{First non-visible point on $\pi$}
\Let{$\pmin$}{$\Ind{\pi}.\textsc{minHeight}((-\infty,\sigma_{\qmin}+\delta],[p+1,\pstop-1]$)} \Comment{Maximizes visibility among reachable points}
\If{$\pmin=\infty$}
\State \Return ``No greedy step possible.'' \Comment{No reachable point has better visibility than $\pi_p$}
\Else
\Let{$\qmin$}{$\Ind{\sigma}.\textsc{minHeight}([\pi_{\pmin}-\delta,\infty),[q,d])$} \Comment{Lowest point on $\sigma$ still seen by $\pmin$}
\State \Return{$\Ind{\pi}.\textsc{maxIndex}((-\infty,\sigma_{\qmin}+\delta],[\pmin,\pstop-1])$}
\EndIf
\EndFunction
  \end{algorithmic}

\caption{Maximal greedy step}

\label{alg:maxgreedy}
\end{algorithm}

\item \textbf{Arbitrary greedy steps $\GreedyStep_{\pi}(\pi',\sigma')$.}
If, in some situation, it is only required to find an arbitrary index
$p'\in\reach_{\pi'}(p,q)$ such that all $p\le i\le p'$ satisfy $\vis_{\sigma'}(i,q)\subseteq\vis_{\sigma'}(p',q)$
or report that no such index exists, we use the function $\GreedyStep_{\pi}(\pi',\sigma')$
to denote that any such function suffices; in particular, $\MinGreedyStep_{\pi}$
or $\MaxGreedyStep_{\pi}$ can be used.
\end{enumerate}
For $\sigma$, we define the obvious symmetric operations. Note that
in these operations, it is not feasible to traverse all directly feasible
points and check whether the visibility criterion is satisfied, since
this would not necessarily yield a running time of $\bigO^{*}(1)$.
\begin{lem}
Using $O((n+m)\log nm)$ preprocessing time, $\MaxGreedyStep_{\pi}$,
$\MinGreedyStep_{\pi}$ and $\stop_\pi$ can be implemented to run in time $\bigO(\log nm)$.\end{lem}
\begin{proof}
In time $\bigO((n+m)\log nm)$, we can build the data structure $\Ind{\pi}$
for $\pi$ and symmetrically $\Ind{\sigma}$ for $\sigma$. Algorithms~\ref{alg:stop},~\ref{alg:smallgreedy}
and~\ref{alg:maxgreedy} implement the greedy steps and $\stop_{\pi}$ using only a constant number of queries to $\Ind{\pi}$ and
$\Ind{\sigma}$, each with running time $\bigO(\log n)$ or $\bigO(\log m)$.
\end{proof}

For the reduced free-space problem, these operations can be implemented even faster.

\begin{lem}
Let $\pi = \pi_{1..n}$ and $\sigma = \sigma_{1..m}$ be input curves of the reduced free-space problem. Using $O((n+m)\log 1/\eps)$ preprocessing time, $\MaxGreedyStep_{\pi}$,
$\MinGreedyStep_{\pi}$ and $\stop_\pi$ can be implemented to run in time $\bigO(\log 1/\eps)$.\end{lem}
\begin{proof}
We argue that range searching can be implemented with $\Oh(\log 1/\eps)$ query time and $\Oh(n \log 1/\eps)$ preprocessing time. This holds since for the point set ${\cal P}=\{(i,\pi_{i})\mid i \in [1 \ldots n]\}$ (1) the $x$-values are $1,\ldots,n$, so that we can determine the relevant pointers in the first level of the fractional-cascading tree in constant time instead of $\Oh(\log n)$ and (2) all $y$-values are multiples of $\tfrac 13 \eps \delta$ and in $[-2\delta,2\delta]$, i.e., there are only $\Oh(1/\eps)$ different $y$-values. For the latter, note that any point $\pi_p > \delta$ sees no point in $\sigma$, and this is preserved by setting $\pi_p$ to $2\delta$ (and similarly for $\sigma$). Using these properties it is straightforward to adapt the fractional-cascading data structure, we omit the details.
\end{proof}

\subsection{Composition of one-dimensional curves}
\label{sec:1d-composition}

In this subsection, we collect essential composition properties of feasible traversals of one-dimensional curves that enable us to tackle the reduced free-space problem (see Figure~\ref{fig:comps} for an illustration of these results). The first tool is a union lemma that states that two intersecting intervals $I, J$ of $\pi$ that each have a feasible traversal together with $\sigma$ prove that also $\pi_{I\cup J}$ can be traversed together with $\sigma$.
\begin{figure}
\centering
\subfloat[Lemma~\ref{lem:composition}]{
  \includegraphics[width=0.31\textwidth]{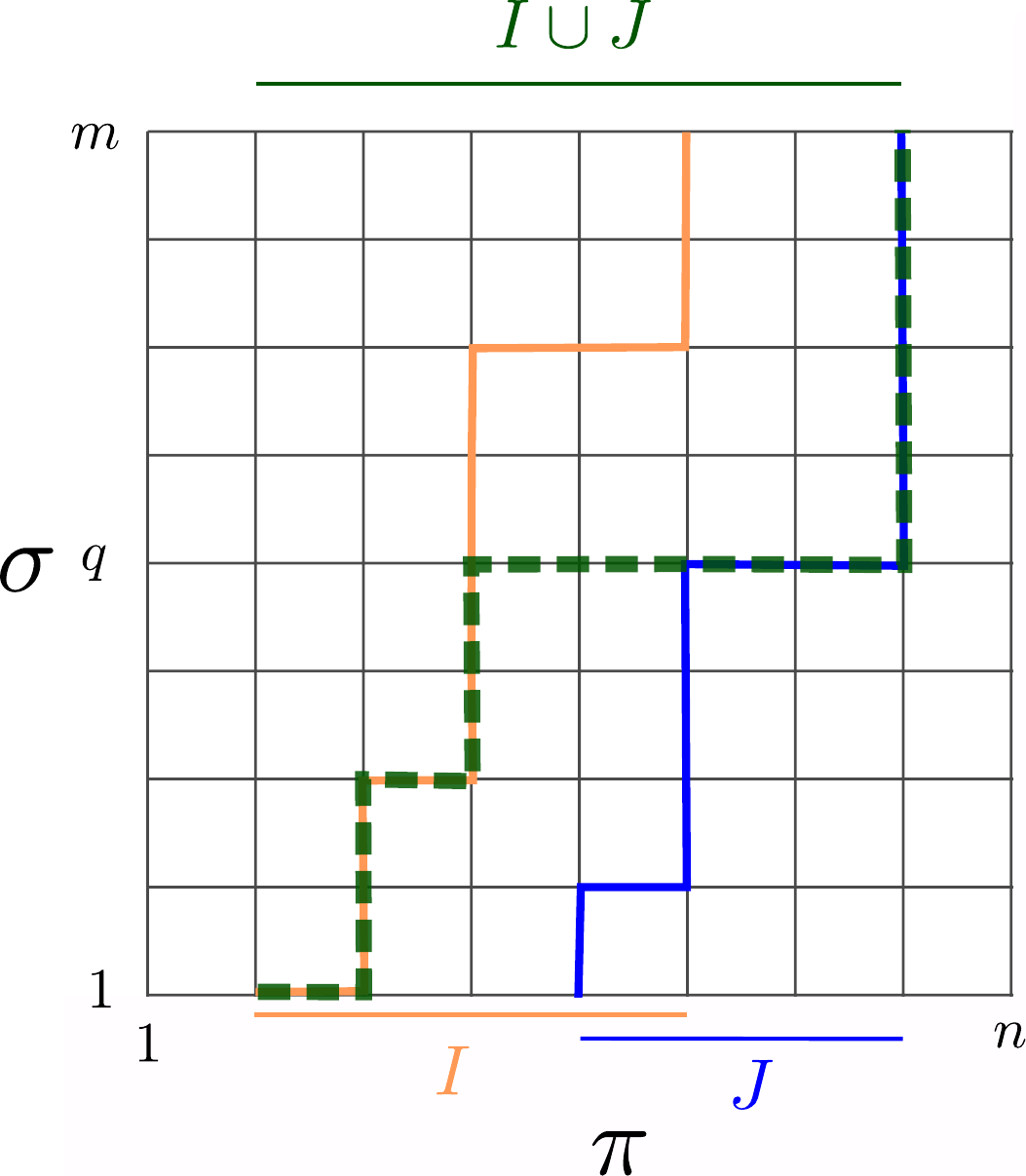}
  \label{fig:compUnion}
}
\,
\subfloat[Lemma~\ref{lem:intersection}]{
  \includegraphics[width=0.31\textwidth]{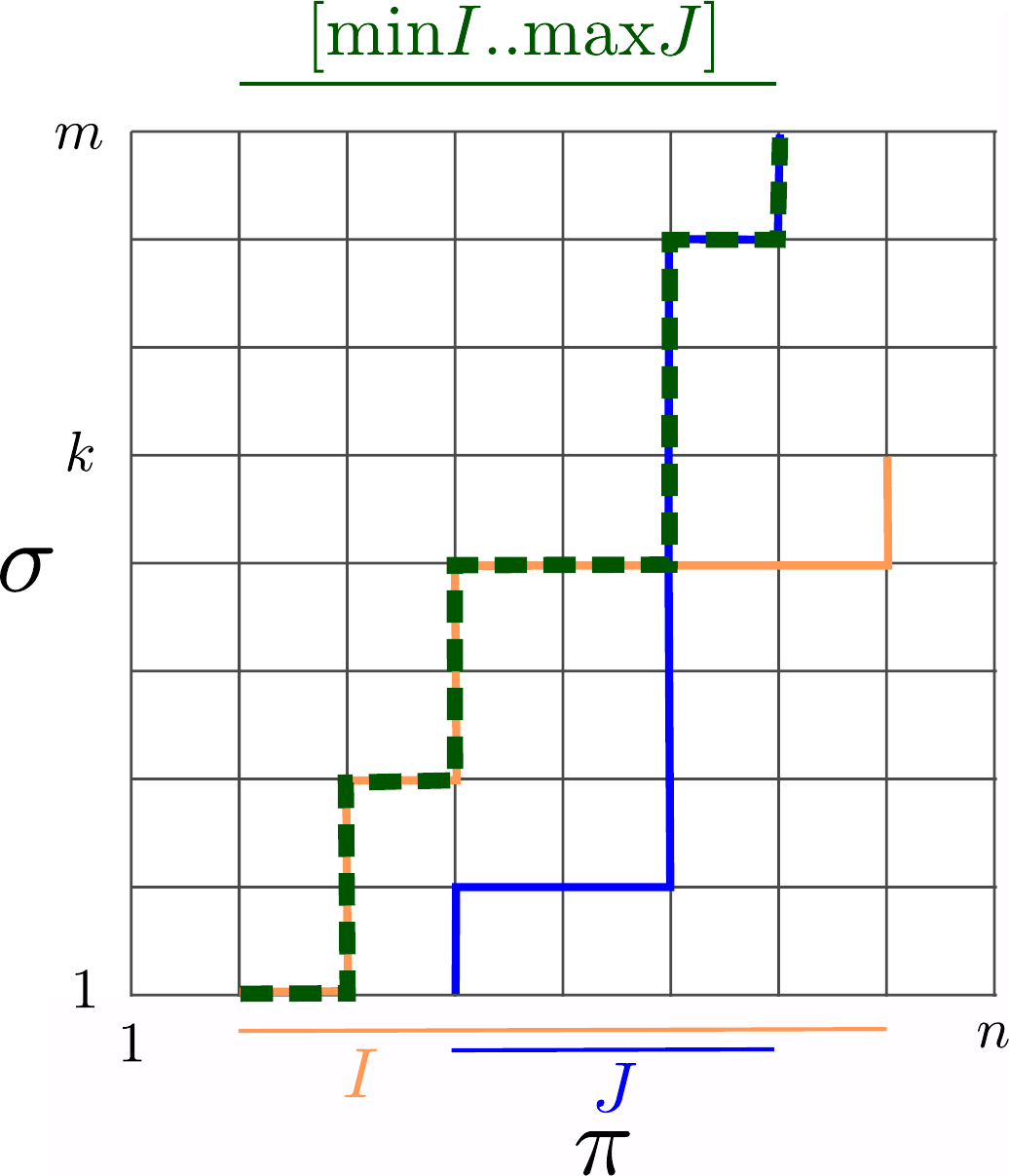}
  \label{fig:compIntersect}
}
\,
\subfloat[Lemma~\ref{lem:composeInitialPiece}]{
  \includegraphics[width=0.31\textwidth]{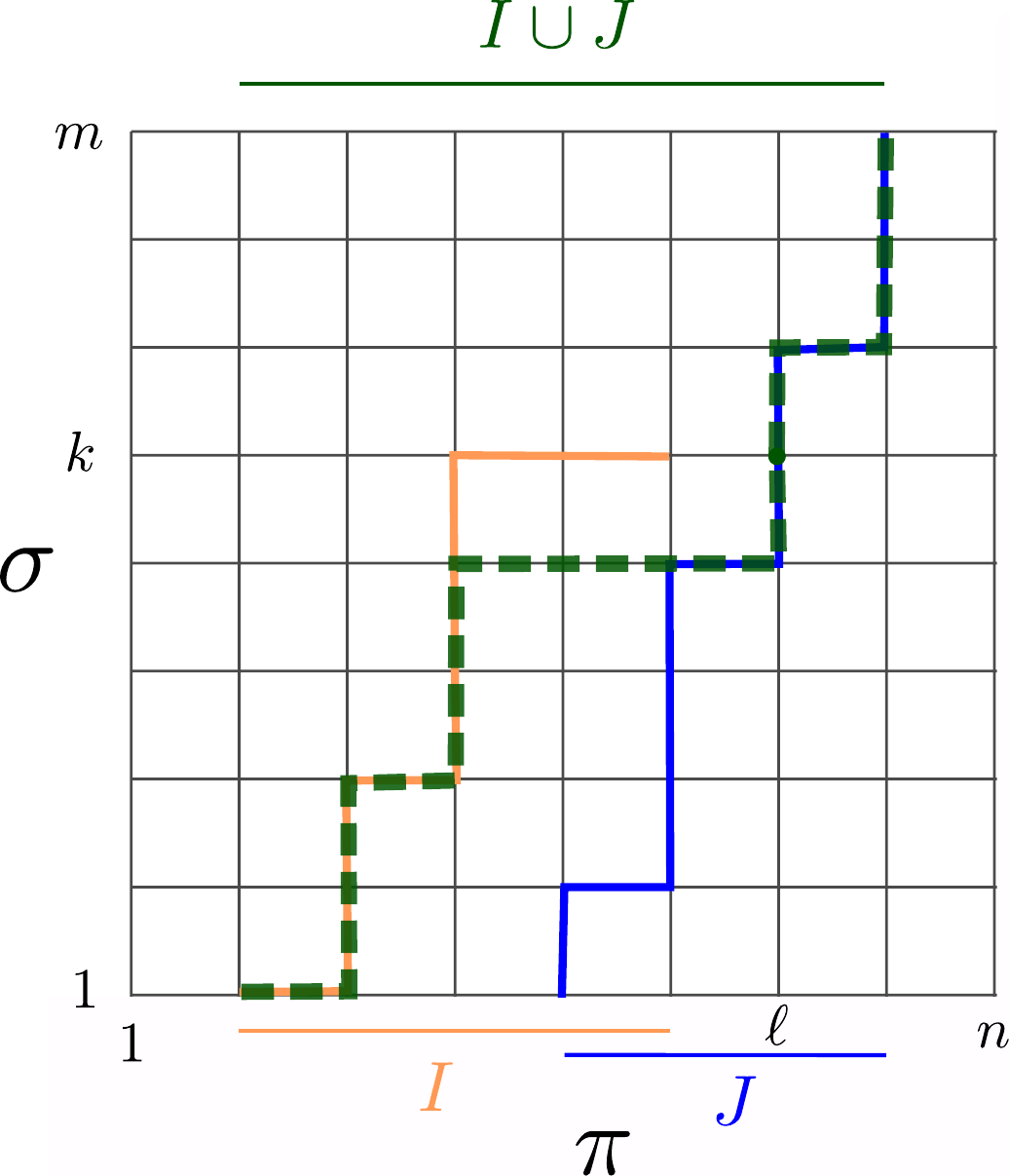}
  \label{fig:compInitial}
}
\caption{Composition properties of feasible traversals of one-dimensional separated curves.}
\label{fig:comps}
\end{figure}

\begin{lem}
\label{lem:composition} Let $\pi = \pi_{1..n}$ and $\sigma = \sigma_{1..m}$
be one-dimensional separated curves and let $I,J\subseteq[1..n]$
be intervals with $I\cap J\ne\emptyset$. If $\ddF(\pi_{I},\sigma)\le \delta$
and $\ddF(\pi_{J},\sigma)\le\delta $, then $\ddF(\pi_{I\cup J},\sigma)\le\delta$.\end{lem}
\begin{proof}
If $I\subseteq J$, the claim is trivial. W.l.o.g, let $I=[a_{I}..b_{I}]$
and $J=[a_{J}..b_{J}]$, where $a_{I}\le a_{J}\le b_{I}\le b_{J}$.
Let $\phi_{I}$ (and $\phi_{J}$) be a feasible traversal of $(\pi_{I},\sigma)$
(and $(\pi_{J},\sigma)$, respectively). By reparameterization,
we can assume that $\phi_{I}(t)=(\psi_{I}(t),f(t))$ and $\phi_{J}(t)=(\psi_{J}(t),f(t))$
for suitable (non-decreasing onto) functions $\psi_{I},\psi_{J}:[0,1]\to[1..n]$ and
$f:[0,1]\to[1..m]$. One of the following cases occurs.

\emph{Case 1:} There is some $0\le t\le1$ with $\psi_{I}(t)=\phi_{J}(t)$.
Then we can concatenate $\phi_{I}(0,t)$ and $\phi_{J}(t,1)$ to obtain
a feasible traversal of $\phi_{I\cup J}$.

\emph{Case 2: }For all $0\le t\le1$, we have $\psi_{I}(t)<\psi_{J}(t)$.
Let $\sigma_{q}$ be the highest point on $\sigma$. By $\ddF(\pi_{I},\sigma)\le \delta$
and $\ddF(\pi_{J},\sigma)\le \delta$, the point $\sigma_{q}$ sees all points on
$\pi_{I\cup J}$. There is some $0\le t^{\ast}\le1$ with $f(t^{\ast})=q$.
We can concatenate $\phi_{I}(0,t)$ and the traversal of $\pi_{\psi_{I}(t^{\ast})..\psi_{J}(t^{\ast})}$
and $\sigma_{q}$ to obtain a feasible traversal of $\pi_{a_{I}..\psi_{I}(t^{\ast})}$
and $\sigma_{1.. f(t^{\ast})}$. Appending $\phi_{J}(t^{\ast},1)$
to this traversal yields $\ddF(\pi_{a_{I}.. b_{J}},\sigma)\le\delta$.
\end{proof}

The second result formalizes situations in which a traversal $\phi$ of subcurves has to cross a traversal $\psi$ of other subcurves, yielding the possibility to follow $\phi$ up to the crossing point and to follow $\psi$ from there on.

\begin{lem}
\label{lem:intersection}
Let $\pi = \pi_{1..n}$ and $\sigma = \sigma_{1..m}$
be one-dimensional curves and consider intervals
$I=[a_{I}..b_{I}]$ and $J=[a_{J}..b_{J}]$ with $J \subseteq I \subseteq [1..n]$,
and $K=[1..k] \subseteq [1..m]$. If $\ddF(\pi_{I},\sigma_{K})\le\delta$
and $\ddF(\pi_{J},\sigma)\le \delta$, then $\ddF(\pi_{a_{I}..b_{J}},\sigma)\le\delta$.\end{lem}
\begin{proof}
Let $\phi$ be a feasible traversal of $\pi_{I}$ and $\sigma_{K}$
and $\psi$ a feasible traversal of $\pi_{J}$ and $\sigma$. We first
show that $\phi$ and $\psi$ cross, i.e., there are $0\le t,t'\le1$
such that $\phi(t)=\psi(t')$. For all $k\in[1.. K]$, let $[s_{k}^{\phi}.. e_{k}^{\phi}]$
denote the interval of points that $\phi$ traverses on $\pi$ while
staying in $\sigma_{k}$. Similarly, $[s_{k}^{\psi}.. e_{k}^{\psi}]$
denotes the interval of points $\psi$ traverses on $\pi$ while staying
in $\sigma_{k}$. 
Assume for contradiction that $[s_{k}^{\phi}.. e_{k}^{\phi}]$
and $[s_{k}^{\psi}.. e_{k}^{\psi}]$ are disjoint for all $1\le k\le K$.
Then initially, we have $s_{1}^{\phi}=a_{I}\le a_{J}=s_{1}^{\psi}$
and hence $e_{1}^{\phi}<s_{1}^{\psi}$. This implies $s_2^\phi \le e_1^\phi+1 \le s_1^\psi \le s_2^\psi$ and inductively we obtain $e_k^\phi < s_k^\psi \le e_k^\psi$ for all $k \in [1..K]$. This contradicts $e_K^\phi = b_I \ge b_J \ge e_K^\psi$. Hence,
for some $1\le k\le K$, $[s_{k}^{\phi}.. e_{k}^{\phi}]$ and $[s_{k}^{\psi}.. e_{k}^{\psi}]$
intersect, which gives $\phi(t)=(p,k)=\psi(t')$ for any $p\in[s_{k}^{\phi},e_{k}^{\phi}]\cap[s_{k}^{\psi},e_{k}^{\psi}]$
and the corresponding $0\le t,t'\le1$. 
By concatenating $\phi(0,t)$ with $\psi(t',1)$, we obtain a feasible
traversal of $\pi_{a_{I}.. b_{J}}$ and~$\sigma$.\end{proof}

The last result in our composition toolbox strengthens \lemref{composition} to the case that the traversal of $\pi_I$ uses only an initial subcurve $\sigma_{1..k}$ of $\sigma$ and not the complete curve.
\begin{lem}
\label{lem:composeInitialPiece}
Let $\pi = \pi_{1..n}$ and $\sigma = \sigma_{1..m}$
be one-dimensional separated curves and consider intervals
$I=[a_{I}.. b_{I}]$ and $J=[a_{J}.. b_{J}]$ with $1\le a_{I}\le a_{J}\le b_{I}\le b_{J}\le n$,
and $K=[1.. k] \subseteq [1..m]$. If $\ddF(\pi_{I},\sigma_{K})\le \delta$
and $\ddF(\pi_{J},\sigma)\le \delta$, then $d_{F}(\pi_{I\cup J},\sigma)\le\delta$.\end{lem}
\begin{proof}
Let $\phi$ be any feasible traversal of $\pi_{J}$ and $\sigma$.
There exists $a_{J}\le\ell\le b_{J}$ with $\phi(t)=(\ell,k)$ for
some $0\le t\le1$. Hence $\phi$ restricted to $[0,t]$ yields a
feasible traversal of $\pi_{a_{J}..\ell}$ and $\sigma_{K}$, i.e.,
$\ddF(\pi_{a_{J}..\ell},\sigma_{K})\le \delta$. Since $I$ and $[a_{J}..\ell]$
are intersecting, Lemma~\ref{lem:composition} yields that $\ddF(\pi_{a_{I}..\ell},\sigma_{K})\le \delta$.
Let $\psi$ be such a feasible traversal of $\pi_{a_{I}..\ell}$
and $\sigma_{K}$. Concatenating $\psi$ at $\psi(1)=(\ell,k)=\phi(t)$
with $\phi(t,1)$, we construct a feasible traversal of $\pi_{a_{I}.. b_{J}}$
and $\sigma$, proving the claim.
\end{proof}

\subsection{Solving the Reduced Free-space Problem}

In this section, we solve the reduced free-space problems, using the structural properties derived in the previous section and the principles underlying the greedy algorithm of \secref{1d-greedydec}. Recall that the greedy steps implemented as discussed in Section~\ref{sub:implementGreedy} run in time $\Oh(\log 1/\eps)$ on the input curves of the reduced free-space problem. 

\subsubsection{Single Entry}
\label{sec:1d-single}

Given the separated curves $\pi=(\pi_{1},\dots,\pi_{n})$ and $\sigma=(\sigma_{1},\dots,\sigma_{m})$
and entry set $E=\{1\}$, we show how to compute $F^{\sigma}$. We
present the following recursive algorithm.

\begin{algorithm}[H]
\begin{algorithmic}[1]
   \Statex

\Function{Find-$\sigma$-exits}{$\pi_{p.. b}, \sigma_{q.. d}$}
\If{$q=d$}
\If{$\stop_\pi(\pi_{p..b},\sigma_q) = b+1$}
\State \Return $\{q\}$ \Comment{The end of $\pi$ is reachable while staying in $\sigma_q$}
\Else{ \Return{$\emptyset$}}
\EndIf
\EndIf
\Statex
\If{$p'\gets \MaxGreedyStep_\pi(\pi_{p.. b}, \sigma_{q.. d})$}
	\State \Return{\textsc{Find-$\sigma$-exits}($\pi_{p'.. b},\sigma_{q .. d})$)}
\ElsIf{$q'\gets \GreedyStep_\sigma(\pi_{p .. b},\sigma_{q.. d})$}
\State \Return{$\textsc{Find-$\sigma$-exits}(\pi_{p .. b},\sigma_{q .. q'-1}) \cup
 \textsc{Find-$\sigma$-exits}(\pi_{p .. b},\sigma_{q' .. d})$}
\Else
\State \Return $\textsc{Find-$\sigma$-exits}(\pi_{p .. b},\sigma_{q .. d-1})$ \Comment{No greedy step possible}
\EndIf
\EndFunction
  \end{algorithmic}

\caption{Special Case: Single entry}

\label{alg:sigmaExits}
\end{algorithm}

The following property establishes that a greedy step on a long curve
is also a greedy step on a shorter curve. Clearly, the converse
does not necessarily hold.
\begin{prop}
\label{prop:greedyToLaterEnd}Let $1\le p\le P\le n$ and $1\le q\le Q\le m$.
Any greedy step on $\pi$ from $(p,q)$ to $(p',q)$ with $p'\le P$
is also a greedy step with respect to $\tilde{\pi}:=\pi_{p.. P}$
and $\tilde{\sigma}:=\sigma_{q.. Q}$, i.e., if there is some $p'\le P$
with $\vis_{\sigma}(i,q)\subseteq\vis_{\sigma}(p',q)$ for all $p\le i\le p'$,
then also $\vis_{\tilde{\sigma}}(i,q)\subseteq\vis_{\tilde{\sigma}}(p',q)$.\end{prop}
\begin{proof}
From the definition of $\vis_{\sigma}$, we immediately derive $\vis_{\tilde{\sigma}}(i,q)=\vis_{\sigma}(i,q)\cap [q..Q]\subseteq\vis_{\sigma}(p',q)\cap[q..Q]=\vis_{\tilde{\sigma}}(p',q)$
for all $p\le i\le p'$. Restricting the length of $\pi$ also has
no influence on the greedy property, except for the trivial requirement
that $p'$ still has to be contained in the restricted curve.\end{proof}
\begin{lem}
Algorithm~\ref{alg:sigmaExits} correctly identifies $F^{\sigma}$
given the single entry $E=\{1\}$.\end{lem}
\begin{proof}
Clearly, if $\textsc{Find-\ensuremath{\sigma}-exits}(\pi,\sigma)$
finds and returns an exit $e$ on $\sigma$, then it is contained
in $F^{\sigma}$, since the algorithm uses only feasible (greedy)
steps. Conversely, we show that for all $I=[p.. b]$ and $J=[q.. d]$,
where $(p,q)$ is a greedy point pair of $\pi$ and $\sigma$, and
all $e\in J$ with $\ddF(\pi_{I},\sigma_{J\cap[1.. e]})\le \delta$,
we have $e\in\FindSigmaExits(\pi_{I},\sigma_{J})$, i.e. we find all
exits. 

Consider some call of $\FindSigmaExits(\pi_{I},\sigma_{J})$ for which
the precondition is fulfilled. If $J$ consists only of a single point,
then $J=\{e\}$, and a feasible traversal of $\pi_{I}$ and $\sigma_{J}$
exists if and only if $\sigma_{e}$ sees all points on $\pi_{I}$.
Let $I=[p.. b],$ then this happens if and only if $\stop_{\pi}(\pi_{I},\sigma_{e})=b+1$,
hence the base case is treated correctly.

Assume that $I=[p.. b]$ and a maximal greedy step $p'$ on $\pi$
exists. By Property~\ref{prop:greedyToLaterEnd}, this step is a
greedy step also with respect to $\sigma_{J\cap[1.. e]}$. Hence
by Lemma~\ref{lem:GreedyCorrectness}, if there is a traversal of
$\pi_{p.. b}$ and $\sigma_{J\cap[1.. e]}$, then a traversal
of $\pi_{[p'.. b]}$ and $\sigma_{J\cap[1.. e]}$ also exists.

Consider the case in which $J=[q.. d]$ and a greedy step $q'$
in $\sigma$ exists. If $e<q'$, then $e\in[q.. q'-1]$ and $J\cap[1.. e]=[q.. q'-1]\cap[1.. e]$.
Hence, $e$ is found in the recursive call with $J'=[q.. q'-1]$.
If $e\ge q'$, then by Property~\ref{prop:greedyToLaterEnd}, this
step is a greedy step with respect to the curves $\pi_{I}$ and $\sigma_{J\cap[1.. e]}$.
Again, by Lemma~\ref{lem:GreedyCorrectness}, the existence of a
feasible traversal of $\pi_{I}$ and $\sigma_{J}$ implies that also
a feasible traversal of $\pi_{I}$ and $\sigma_{J\cap[q'.. e]}$
exists.

It remains to regard the case in which no greedy step exists. By Lemma~\ref{lem:GreedyCorrectness},
there is no feasible traversal of $\pi_{1..n}$ and $\sigma_{1..d}$. This implies $e\ne d$ and all exits are found in the recursive
call with $J'=[q,d-1]$.\end{proof}
\begin{lem}
\label{lem:findSigmaExitsTime}$\FindSigmaExits(\pi_{p.. b},\sigma_{q.. d})$
runs in time $\bigO((d-q+1)\cdot\log 1/\eps)$.\end{lem}
\begin{proof}
Since the algorithm's greedy steps on $\pi$ are maximal, after each
greedy step on $\pi$, we split $\sigma$ (by a greedy step on $\sigma$)
or shorten $\sigma$ (if no greedy step on $\sigma$ is found). Thus,
it takes at most $\bigO(\log 1/\eps)$ time until $\sigma$ is split or
shortened. The base case is also handled in time $\bigO(\log 1/\eps)$.
In total, this yields a running time of $\bigO((d-q+1)\log 1/\eps)$.
\end{proof}
Note that by swapping the roles of $\pi$ and $\sigma$, $\FindSigmaExits$
can be used to determine $F^{\pi}$ given the single entry $\sigma_{1}$
on $\sigma$. This is equivalent to having the single entry $E=\{1\}$
on $\pi$. Thus, we can also implement the function $\FindPiExits(\pi_{1.. n},\sigma_{1.. m})$
that returns $F^{\pi}$ given the single entry $E=\{1\}$ in
time $\bigO(n\log 1/\eps)$.

\subsubsection{Entries on $\pi$, Exits on $\pi$}

\label{sec:1d-multiPi}

In this section, we tackle the task of determining $F^{\pi}$ given
a set of entries $E$ on $\pi$. It is essential to avoid computing
the exits by iterating over every single entry. We show how to divide
$\pi$ into disjoint subcurves that can be solved by a single call
to $\FindPiExits$ each.

Assume we want to traverse $\pi_{p.. b}$ and $\sigma_{q.. d}$
starting in $\pi_{p}$ and $\sigma_{q}$. Let $u(p):=\max\{p'\in[p,b]\mid\exists q\le q'\le d:\ddF(\pi_{p.. p'},\sigma_{q.. q'})\le \delta\}$ 
be the last point on $\pi$ that is reachable while traversing an
arbitrary subcurve of $\sigma_{q.. d}$ that starts in $\sigma_{q}$.
This point fulfills the following properties.
\begin{lem}
\label{lem:propup}It holds that
\begin{enumerate}
\item If there are $p\le e\le e'\le u(p)$ with $\ddF(\pi_{e.. e'},\sigma_{q.. d})\le\delta$,
then $\ddF(\pi_{p.. e'},\sigma_{q.. d})\le\delta$.
\item For all $p\le e\le u(p)<e'$, we have that $\ddF(\pi_{e.. e'},\sigma_{q.. d})>\delta$.
\end{enumerate}
\end{lem}
\begin{proof}
By definition of $u(p)$, there is a $q\le q'\le d$ with $\ddF(\pi_{p.. u(p)},\sigma_{q.. q'})\le \delta$.
Since $[e,e']\subseteq[p,u(p)]$, Lemma~\ref{lem:intersection} proves
the first statement. For the second statement, assume for contradiction
that $\ddF(\pi_{e.. e'},\sigma_{q.. d})\le\delta$. Then, Lemma~\ref{lem:composeInitialPiece}
yields that $\ddF(\pi_{p.. e'},\sigma_{q.. d})\le \delta$. This is
a contradiction to the choice of $u(p)$, since $e'>u(p)$.
\end{proof}
The above lemma implies that we can ignore all entries in $[p.. u(p)]$
except for $p$ and that all exits reachable from $p$ are contained
in the interval $[p.. u(p)]$. This gives rise to the following
algorithm.

\begin{algorithm}[H]
\begin{algorithmic}[1]
   \Statex

\Function{$\pi$-exits-from-$\pi$}{$\pi, \sigma, E$}
\State $S \gets \emptyset$
\While{$E\ne \emptyset$}
  \Let{$\hat{p}$}{pop minimal index from $E$}
  \State $p \gets \hat{p}$, $q\gets 1$
  \Repeat
	\If{$q'\gets\MaxGreedyStep_\sigma(\pi_{p.. n},\sigma_{q.. m})$}
	\Let{$q$}{$q'$}
	\EndIf
	\If{$p'\gets\GreedyStep_\pi(\pi_{p.. n},\sigma_{q.. m})$}
	\Let{$p$}{$p'$}
	\EndIf
  \Until{no greedy step was found in the last iteration}
  \Let{$\bar{p}$}{$\stop_\pi(\pi_{p.. n},\sigma_{q.. m})- 1$} \Comment{determines the maximal reachable point $u(\hat{p})$}
  \Let{$S$}{$S\cup \FindPiExits(\pi_{\hat{p}..\bar{p}},\sigma)$}
  \Let{$E$}{$E \cap [\bar{p}+1,n]$} \Comment{drops all entries in $[\hat{p},u(\hat{p})]$}
\EndWhile
\State \Return $S$
\EndFunction
  \end{algorithmic}

\caption{Given entry points $E$ on $\pi$, compute all exits on $\pi$.}

\label{alg:piExitspiEntries}
\end{algorithm}

\begin{lem}
Algorithm~\ref{alg:piExitspiEntries} correctly computes $F^{\pi}$.\end{lem}
\begin{proof}
We first argue that for each considered entry $\hat{p}$, the algorithm
computes $\bar{p}=u(\hat{p})$. Clearly, $\bar{p}\le u(\hat{p})$,
since only feasible steps are used to reach $\bar{p}$. If $\bar{p}=m$,
this already implies that also $u(\hat{p})=m$. Otherwise, let $(p,q)$
be the greedy point pair on the curves $\pi_{\hat{p} \ldots n}$ and $\sigma$
for which no greedy step has been found. Then by Lemma~\ref{lem:GreedyCorrectness},
for $\pstop:=\stop_{\pi}(\pi_{p.. n},\sigma_{q.. m})$ and all
$1\le q'\le m$, we have that $\ddF(\pi_{\hat{p}..\pstop},\sigma_{1.. q'})>\delta$.
Hence, $u(\hat{p})<\pstop$. Finally, note that Algorithm~\ref{alg:piExitspiEntries}
computes $\bar{p}=\pstop-1$, which proves $\bar{p}=u(\hat{p})$.

It is clear that every found exit is included in $F^{\pi}$. Conversely,
let $e'\in F^{\pi}$ and $1\le e\le n$ be such that $\ddF(\pi_{e.. e'},\sigma)\le\delta$.
For some $\hat{p}$ with $\hat{p}\le e\le u(\hat{p})=\bar{p}$, we
run $\FindPiExits(\pi_{\hat{p}..\bar{p}},\sigma)$. Hence by Lemma~\ref{lem:propup}~(2),
$e'\le u(\hat{p})$ and by Lemma \ref{lem:propup}~(1), $\ddF(\pi_{\hat{p}.. e'},\sigma)\le\delta$.
Hence, the corresponding call $\FindPiExits(\pi_{\hat{p}..\bar{p}},\sigma)$
will find $e'$.\end{proof}
\begin{lem}
Using preprocessing time $\bigO((n+m)\log 1/\eps)$, Algorithm~\ref{alg:piExitspiEntries}
runs in time $\bigO(n\log 1/\eps)$.\end{lem}
\begin{proof}
We first bound the cost of all calls $\FindPiExits(\pi_{I_{i}},\sigma)$.
Clearly, all intervals $I_{i}$ are disjoint with $\bigcup I_{i}\subseteq[1.. n]$.
Hence, by Lemma~\ref{lem:findSigmaExitsTime}, the total time spent
in these calls is bounded by $\bigO(\sum_{i}|I_{i}|\log(1/\eps))=\bigO(n\log 1/\eps)$.
To bound the number greedy steps, let $p_{1},\dots,p_{k}$ be the
distinct indices considered as values of $p$ during the execution
of $\PiExitsFromPi(\pi,\sigma)$. Between changing $p$ from each
$p_{i}$ to $p_{i+1}$, we will make, by maximality, at most one call
to $\MaxGreedyStep_{\sigma}$ and at most call to $\GreedyStep_{\pi}$.
Since $k\le n,$ the total cost of greedy calls is bounded by $\bigO(n\log 1/\eps)$
as well. The total time spent in all other operations is bounded by
$\bigO(n\log 1/\eps)$.
\end{proof}

\subsubsection{Entries on $\pi$, Exits on $\sigma$}

\label{sec:1d-multiSigma}

Similar to the previous section, we show how to compute the exits $F^{\sigma}$
given entries $E$ on $\pi$, by reducing the problem to calls of
$\FindSigmaExits$ on subcurves of $\pi$ and $\sigma$. This time,
however, the task is more intricate. For any index $p$ on $\pi$,
let $Q(p):=\min\{q\mid\ddF(\pi_{p.. n},\sigma_{1.. q})\le\delta\}$
be the endpoint of the shortest initial fragment of $\sigma$ such
that the remaining part of $\pi$ can be traversed together with this
fragment\footnote{As a convention, we use $\min \emptyset = \max \emptyset=\infty$.}. Let $P(p):=\min\{p'\mid\ddF(\pi_{p.. p'},\sigma_{1.. Q(p)})\le\delta\}$
be the endpoint of the shortest initial fragment of $\pi$, such that
$\sigma_{Q(p)}$ can be reached by a feasible traversal.

Note that by definition, entries $p$ with $Q(p)=\infty$ are irrelevant
for determining the exits on $\sigma$. In fact, if an entry $p$
is \emph{relevant}, i.e., $Q(p)<\infty$, it is easy to compute $Q(p)$
due to the following lemma.
\begin{lem}
\label{lem:QpCharact}
Let $Q'(p):=\min\{q\mid\sigma_{q}\ge\max_{i\in[p.. n]}\pi_{i}-\delta\}$.
If $Q(p)<\infty$, then $Q(p)=Q'(p)$. Similarly, $Q(p)<\infty$ implies that 
$P(p)=\min\{p' \mid \pi_{p'} \le \min_{i\in[q..Q(p)]} \sigma_i + \delta\}<\infty$.
\end{lem}
\begin{proof}
Assume that $Q(p)<Q'(p)$ holds, then no point in $\sigma_{1.. Q(p)}$
sees the highest point in $\pi_{p.. n}$. Hence no feasible traversal
of these curves can exist, yielding a contradiction. Assume that $Q(p)>Q'(p)$
holds instead and consider the feasible traversal $\phi$ of the shortest
initial fragment of $\sigma$ that passes through all points in $\pi_{p.. n}$.
At some point $\phi$ visits $(\pi_{p'},\sigma_{Q'(p)})$ for some
$p\le p'\le n$. We can alter this traversal to pass through the remaining
curve $\pi_{p'.. n}$ while staying in $\sigma_{Q'(p)}$, since
$\sigma_{Q'(p)}$ sees all points on $\pi_{p'.. n}$. This gives
a feasible traversal of $\pi_{p.. n}$ and $\sigma_{1.. Q'(p)}$,
which is a contradiction to the choice of $\phi$ and $Q(p)>Q'(p)$.

The statement for $P(p)$ follows analogously by regarding the curves $\pi_{p..n}$ and $\sigma_{1..Q(p)}$ and switching their roles.
\end{proof}
Note that the previous lemma shows that for relevant entries $p_{1}<p_{2}$,
we have $Q(p_{1})\ge Q(p_{2})$, since for relevant entries, $Q(p_{1})=Q'(p_{1})\ge Q'(p_{2})=Q(p_{2})$.
We will use the following lemma to argue that (i) if $Q(p_{1})=Q(p_{2})$,
entry $p_{1}$ dominates $p_{2}$, and (2) if $Q(p_{1})>Q(p_{2})$,
we have $p_{2}\notin[p_{1}.. P(p_{1})]$. Hence, we can ignore all
entries in $[p_{1}.. P(p_{1})]$ except for $p_{1}$ itself. 

\begin{figure}
\centering
\subfloat[Case $q_1 =q_2$.]{
  \includegraphics[width=0.45\textwidth]{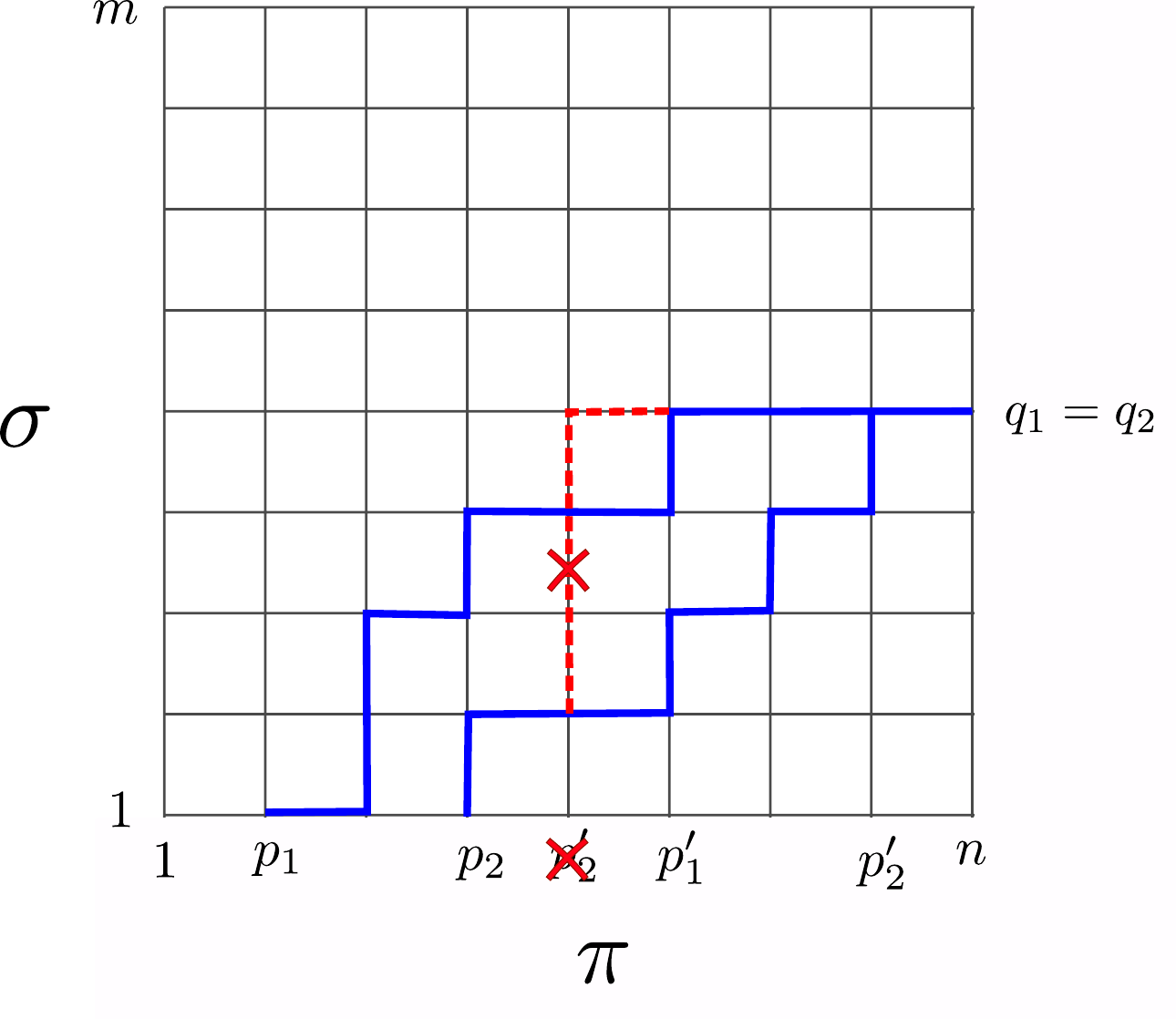}
  \label{fig:lemnoSkipA}
}
\quad
\subfloat[Case $q_1 > q_2$.]{
  \includegraphics[width=0.45\textwidth]{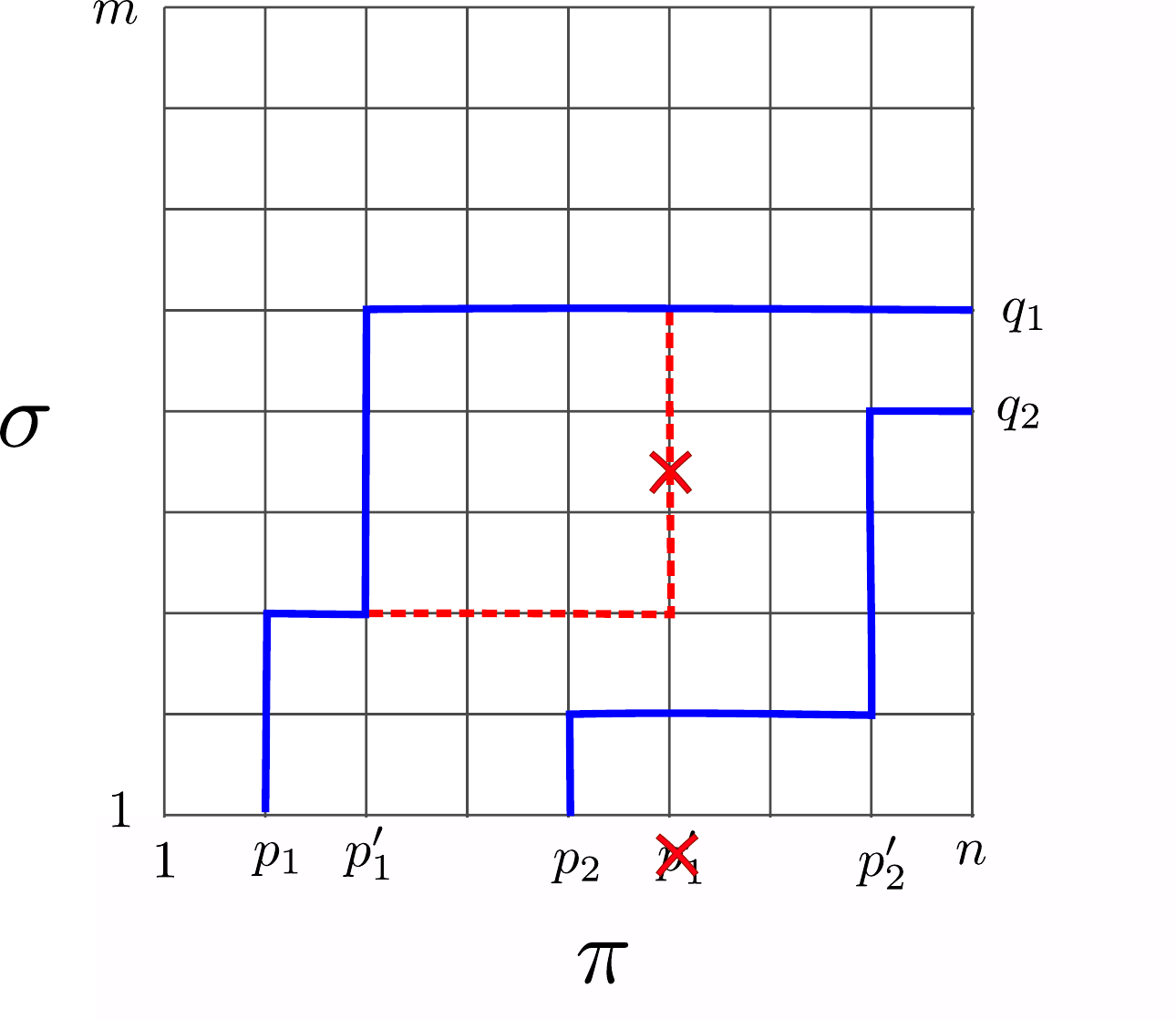}
  \label{fig:lemnoSkipB}
}
\caption{Illustration of Lemma~\ref{lem:noSkip}. For both $p_i$, $i \in \{1,2\}$, a feasible traversal of the curves $\pi_{p_i..p_i'}$ and $\sigma_{1..q_i}$ is depicted as monotone paths in the free-space.}
\label{fig:lemnoSkip}
\end{figure}

\begin{lem}
\label{lem:noSkip}Let $p_{1}<p_{2}$ be indices on $\pi$ with $q_{1}:=Q(p_{1})<\infty$
and $q_{2}:=Q(p_{2})<\infty$. Let $p'_{1}:=P(p_{1})$ and $p'_{2}:=P(p_{2})$.
If $q_{1}=q_{2}$, then $p_{1}'\le p_{2}'$. Otherwise, i.e., if $q_{1}>q_{2}$,
we even have $p'_{1}<p_{2}$. \end{lem}
\begin{proof}
See \figref{lemnoSkip} for illustrations.
Let $q_{1}=q_{2}$. Assume for contradiction that $p'_{1}>p_{2}'$,
then we have $\ddF(\pi_{p_{1}.. p_{1}'},\sigma_{1.. q_{1}})\le\delta$
and $\ddF(\pi_{p_{2}.. p_{2}'},\sigma_{1.. q_{1}})\le\delta$, where
$[p_{2}.. p_{2}']\subseteq[p_{1}.. p_{1}']$. Hence by Lemma~\ref{lem:intersection},
$\ddF(\pi_{p_{1}.. p_{2}'},\sigma_{1.. q_{1}})\le\delta$ and thus
$p_{1}'\le p_{2}'$, which is a contradiction to the assumption. 

For the second statement, let $p$ be maximal such that $\pi_{p}>\sigma_{q_{2}}+\delta$.
If $p$ does not exist or $p<p_{1}$, we have that $Q'(p_{1})=Q'(p_{2})$
and hence by Lemma~\ref{lem:QpCharact}, $q_{1}=q_{2}$. Note that
additionally $p<p_{2}$, since otherwise $\sigma_{q_{2}}<\pi_{p}-\delta$
with $p\ge p_{2}$ shows that $q_{2}\ne Q'(p_{2})$ contradicting
Lemma~\ref{lem:QpCharact}. Thus, in what follows, we can assume
that $p_{1}<p<p_{2}$.

Assume for contradiction that $q_{1}>q_{2}$ and $p'_{1}\ge p_{2}$.
Then a feasible traversal $\phi$ of $\pi_{p_{1}.. p_{1}'}$ and
$\sigma_{1.. q_{1}}$ visits $(\pi_{p},\sigma_{q})$ for some $1\le q\le q_{1}$.
It even holds that $q<q_{1}$, since otherwise there is a feasible
traversal of $\sigma_{1.. q_{1}}$ and $\pi_{p_{1}.. p}$ with
$p<p'_{1}$, contradicting the choice of $p_{1}'$. Clearly, $\sigma_{q}>\sigma_{q_{2}}$,
since $\pi_{p}$ sees $\sigma_{q}$, while it does not see $\sigma_{q_{2}}$.
Since by choice of $p$, $\sigma_{q_{2}}$ sees all of $\pi_{p+1.. n}$
and $\sigma_{q}$ sees only more (including $\pi_{p}$), we conclude
that we can traverse all points of $\pi_{p.. n}$ while staying
in $\sigma_{q}$. Concatenating this traversal to the feasible traversal
$\phi$ yields $\ddF(\pi_{p_{1}.. n},\sigma_{1.. q})\le \delta$ and
thus $Q'(p_{1})\le q<q_{1}$, which is a contradiction to Lemma~\ref{lem:QpCharact}.
This proves that $q_{1}>q_{2}$ implies $p_{1}'\le p_{2}$.
\end{proof}

\begin{algorithm}
\begin{algorithmic}[1]
   \Statex

\Function{$\sigma$-exits-from-$\pi$}{$\pi, \sigma, E$}
\State $F \gets \emptyset$, $\bar{q} \gets m$
\Repeat
  \Let{$\hat{p}$}{pop minimal index from $E$}
  \State $p \gets \hat{p}$, $q\gets 1$
  \State $Q' \gets Q'(p)$
  \Repeat
	\If{$q'\gets\MaxGreedyStep_\sigma(\pi_{p.. n},\sigma_{q.. Q'})$}
			\State $q \gets q'$
	\EndIf
	\If{$q\ne Q'$ and $p'\gets\MinGreedyStep_\pi(\pi_{p.. n},\sigma_{q.. Q'})$}
	\Let{$p$}{$p'$}
	\EndIf
  \Until{$q = Q'$ or no greedy step was found in the last iteration}
  \If{$q = Q'$}
	\Let{$F$}{$F\cup \FindSigmaExits(\pi_{p.. n},\sigma_{Q'.. \bar{q}})$}
	\Let{$\bar{q}$}{$Q'-1$}
  \EndIf
  \Let{$E$}{$E \cap [p+1,n]$}
\Until{$E=\emptyset$}
\State \Return $F$
\EndFunction
  \end{algorithmic}

\caption{Given entry points $E$ on $\pi$, compute all exits on $\sigma$.}

\label{alg:SigmaExitspiEntries}
\end{algorithm}

\begin{lem}
\label{lem:algProps}Algorithm~\ref{alg:SigmaExitspiEntries} fulfills
the following properties.
\begin{enumerate}
\item Let $(p,q)$ with $q<Q'(\hat{p})$ be a greedy point pair of $\pi_{\hat{p}.. n}$
and $\sigma_{1.. Q'(\hat{p})}$ for which no greedy step exists.
For all $e\in[\hat{p},p]$, we have $Q(e)=\infty$.
\item For each $\hat{p}$ considered, if $Q(\hat{p})<\infty$, the algorithm
calls $\mbox{\ensuremath{\FindSigmaExits}}(\pi_{P(\hat{p}).. n},\sigma_{Q(\hat{p})..\bar{q}})$.
In this case, the point $(P(\hat{p}),Q(\hat{p}))$ is a greedy pair
of $\pi_{\hat{p}.. n}$ and $\sigma$.
\end{enumerate}
\end{lem}

\begin{proof}
For the first statement, assume for contradiction that $Q(e)<\infty$.
By Lemma~\ref{lem:QpCharact}, $Q(e)=Q'(e)$, which implies that
for all $q'<Q'(e)\le Q'(\hat{p})$, we have $\sigma_{q'}<\sigma_{Q'(e)}$ and
hence $\vis_{\pi}(p,q')\subseteq\vis_{\pi}(p,Q'(e))$. Hence, $\stop_{\sigma}(\pi_{p.. n},\sigma_{q.. Q'(\hat{p})})\le Q'(e)$,
since otherwise $Q'(e)\gets\GreedyStep_{\sigma}(\pi_{p.. n},\sigma_{q.. Q'(\hat{p})})$.
By Lemma~\ref{lem:GreedyCorrectness}, this proves that $\ddF(\pi_{\hat{p}.. n},\sigma_{1.. Q'(e)})> \delta$. 
Since $\ddF(\pi_{\hat{p}..e},\sigma_{1..q'})\le \delta$ for some $q'<Q'(e)$, Lemma~\ref{lem:composeInitialPiece} yields $\ddF(\pi_{e..n},\sigma_{1..Q'(e)})>\delta$. 
This is a contradiction to $Q(e)=Q'(e)$.

For the second statement, note that if $Q(\hat{p})<\infty$, then
by Lemma~\ref{lem:QpCharact}, $Q(\hat{p})=Q'(\hat{p})$. Hence Lemma~\ref{lem:GreedyCorrectness}
yields that the algorithm finds a feasible traversal of $\pi_{\hat{p}.. p}$
and $\sigma_{1.. Q'(\hat{p})}$ for some $\hat{p}\le p\le n$.
This shows that $P(\hat{p})\le p<\infty$. Let $\sigma':=\sigma_{1.. Q(\hat{p})}$
and assume that there is a $p'<p$ with $\ddF(\pi_{\hat{p}.. p'},\sigma')\le \delta$
and let $(\tilde{p},\tilde{q})$ be the greedy point of $\pi_{\hat{p}.. n}$
and $\sigma'$ right before the algorithm made a greedy step on $\pi$
to some index in $(p',p]$. By maximality of the greedy steps on $\sigma$,
there exists $\tilde{q}<\qmin<Q(\hat{p})$ such that $\pi_{\tilde{p}}$
does not see $\sigma_{\qmin}$, since otherwise $Q(\hat{p})\in\reach_{\sigma'}(\tilde{p},\tilde{\sigma})$
with $\vis_{\pi}(\tilde{p},\tilde{q})\subsetneq\vis_{\pi}(\tilde{p},Q(\hat{p}))$, i.e., $Q(\hat{p})$ 
would be a greedy step on $\sigma'$. By minimality of greedy steps
on $\pi$, $\vis_{\sigma'}(\tilde{p},\tilde{q})\supsetneq\vis_{\sigma'}(i,\tilde{q})$
for all $\tilde{p}\le i\le p'$. Hence, no vertex on $\pi_{\tilde{p}.. p'}$
sees $\sigma_{\qmin}$, which proves $\ddF(\pi_{\tilde{p}.. p'},\sigma')> \delta$.
Since $(\tilde{p},\tilde{\sigma})$ is a greedy pair of $\pi_{\hat{p}.. p'}$
and $\sigma'$, this yields that $\ddF(\pi_{\hat{p}.. p'},\sigma')> \delta$
by Lemma~\ref{lem:GreedyCorrectness}, which is a contradiction to
the assumption. Hence, the algorithm calls $\FindSigmaExits(\pi_{p.. n},\sigma_{Q'(\hat{p})..\bar{q}})$,
where $p=P(\hat{p})$ and $Q'(\hat{p})=Q(\hat{p})$.

It remains to show that $(P(\hat{p}),Q(\hat{p}))$ is also a greedy
pair of $\pi_{\hat{p}.. n}$ and the \emph{complete} curve $\sigma$.
By Lemma~\ref{lem:QpCharact}, every $\hat{p}\le p<P(\hat{p})$ satisfies $\pi_{p}>\pi_{P(\hat{p})}$
and hence $\vis_{\sigma}(p,q)\subseteq\vis_{\sigma}(P(\hat{p}),q)$
for all $1\le q\le m$. Hence, if at some greedy pair $(p,q)$, $q\le Q(\hat{p})$, a greedy
step $p'\gets\GreedyStep_{\pi}(\pi_{p.. n},\sigma)$ with $p'\ge P(\hat{p})$
exists, then also $P(\hat{p})\gets\GreedyStep_{\pi}(\pi_{p.. n},\sigma)$,
which shows that $(P(\hat{p}),q)$ is a greedy point of $\pi_{\hat{p}.. n}$
and $\sigma$. If $q=Q(\hat{p})$, then $(P(\hat{p}),Q(\hat{p}))$ is a greedy point pair. Otherwise, by Lemma~\ref{lem:QpCharact}, $P(\hat{p})$ sees
all of $\sigma_{q.. Q(\hat{p})}$ and $\sigma_{q}<\sigma_{Q(\hat{p})}$,
hence $Q(\hat{p})\in\GreedyStep_{\sigma}(\pi_{P(p).. n},\sigma)$
and $(P(\hat{p}),Q(\hat{p}))$ is a greedy step of $\pi_{\hat{p}.. n}$
and $\sigma$.

It is left to consider the case that for all greedy pairs $(p,q)$, $q\le Q(\hat{p})$,
of $\pi_{\hat{p}.. n}$ and $\sigma$, no greedy step to some $p'\ge P(\hat{p})$
exists. Then there is some $(p,q)$ with $p<P(\hat{p})$ and $q\le Q(\hat{p})$ for which
no greedy step exists at all. We have $\pstop:=\stop_{\pi}(\pi_{p.. n},\sigma_{q.. m})\le P(\hat{p})$,
since otherwise $P(\hat{p})$ would be a greedy step. Since Lemma~\ref{lem:GreedyCorrectness}
shows that $\ddF(\pi_{\hat{p}.. \pstop},\sigma_{1..q})> \delta$, this contradicts $\ddF(\pi_{\hat{p}.. P(\hat{p})},\sigma_{1.. Q(\hat{p})})\le \delta$.

\end{proof}
\begin{lem}
Algorithm~\ref{alg:SigmaExitspiEntries} correctly computes $F^{\sigma}$.\end{lem}
\begin{proof}
Clearly, any exit found is contained in $F^{\sigma}$, since $\SigmaExitsFromPi$
and $\FindSigmaExits$ only use feasible steps. For the converse,
let $e\in E$ be an arbitrary entry and consider the set $F_{e}^{\sigma}=\{q\mid\ddF(\pi_{e.. n},\sigma_{1.. q})\le\delta\}$
of $\sigma$-exits corresponding to the entry $e$. 

We first show that if $F_{e}^{\sigma}\ne\emptyset$ and hence $Q(e)<\infty,$
we have $F_{e}^{\sigma}=\FindSigmaExits(\pi_{P(e).. n},\sigma_{Q(e).. m})$.
Let $\bar{e}\in F_{e}^{\sigma}$. By Lemma~\ref{lem:algProps}, $(P(e),Q(e))$
is a greedy pair of $\pi_{e.. n}$ and $\sigma$ and hence also
of $\pi_{e.. n}$ and $\sigma_{1..\bar{e}}$. Lemma~\ref{lem:GreedyCorrectness}
thus implies $\ddF(\pi_{P(e).. n},\sigma_{Q(e)..\bar{e}})\le \delta$
and consequently $\bar{e}\in\FindSigmaExits(\pi_{P(e).. n},\sigma_{Q(e).. m})$.
The converse clearly holds as well.

Note that $e$ is not considered as $\hat{p}$ in any iteration of
the algorithm if and only if the algorithm considers some $\hat{p}$
with $e\in[\hat{p}+1.. p]$, where either (i) the algorithm finds
a greedy pair $(p,q)$ of $\pi_{\hat{p}.. n}$ and $\sigma_{1.. Q'(\hat{p})}$
that allows no further greedy steps, or (ii) the algorithm calls $\FindSigmaExits(\pi_{p.. n},\sigma_{Q'(\hat{p})..\bar{q}})$,
where $p=P(\hat{p})$ by Lemma~\ref{lem:algProps}. In the first
case, $F_{e}^{\sigma}=\emptyset$ since Lemma~\ref{lem:algProps}
proves $Q(e)=\infty$. In the second case, if $F_{e}^{\sigma}\ne\emptyset$,
we have $Q(e)<\infty$, and hence by Lemma~\ref{lem:noSkip}, $Q(e)=Q(\hat{p})$
and $P(\hat{p})\le P(e)$. Since $\sigma_{Q(\hat{p})}$ sees all of
$\pi_{P(\hat{p}).. n}$, any exit reachable from $(P(e),Q(e))$
is reachable from $(P(\hat{p}),Q(\hat{p}))$ as well. Hence $F_{e}^{\sigma}\subseteq F_{\hat{p}}^{\sigma}$. 

Let $\hat{p}_{1}\le..\le\hat{p}_{k}$ be the entries considered
as $\hat{p}$ by the algorithm. It remains to show that the algorithm
finds all exits $\bigcup_{i=1}^{k}F_{\hat{p}_{i}}^{\sigma}$. We inductively
show that the algorithm computes $F_{\hat{p}_{i}}^{\sigma}\setminus\bigcup_{j<i}F_{\hat{p}_{j}}^{\sigma}$
in the loop corresponding to $\hat{p}=\hat{p}_i$. The base
case $i=1$ follows immediately. Note that for every $i\ge2$, the
corresponding loop computes $\FindSigmaExits(\pi_{P(\hat{p}_{i}).. n},\sigma_{Q(\hat{p}_{i}).. Q(\hat{p}_{i-1})-1})=F_{\hat{p}_{i}}^{\sigma}\cap[Q(\hat{p}_{i}).. Q(\hat{p}_{i-1})-1]$.
The claim follows if we can show $F_{\hat{p}_{i}}^{\sigma}\cap[Q(\hat{p}_{i-1}).. m]\subseteq F_{\hat{p}_{i-1}}$.
Let $\bar{e}\in F_{\hat{p}_{i}}^{\sigma}$ with $\bar{e}\ge Q(\hat{p}_{i-1})$.
Then $\ddF(\pi_{\hat{p_{i}}.. n},\sigma_{1..\bar{e}})\le \delta$.
Together with $\ddF(\pi_{\hat{p}_{i-1}.. n},\sigma_{1.. Q(\hat{p}_{i-1})})\le \delta$,
Lemma~\ref{lem:intersection} shows that $\ddF(\pi_{\hat{p}_{i-1}.. n},\sigma_{1..\bar{e}})\le \delta$
and hence $e\in F_{\hat{p}_{i-1}}^{\sigma}$.

\end{proof}
\begin{lem}
Algorithm~\ref{alg:SigmaExitspiEntries} runs in time $\bigO((n+m)\log 1/\eps)$.\end{lem}
\begin{proof}
Consider the total cost of the calls $\FindSigmaExits(\pi_{I_i}, \sigma_{J_i})$. Since all $J_i$ are disjoint and $\bigcup_i J_i \subseteq [1..m]$, \lemref{findSigmaExitsTime} bounds the total cost of such calls by $\Oh(\sum_i |J_i| \log(1/\eps))=\Oh(m\log(1/\eps))$. Let $p_1,\dots,p_k$ denote the distinct indices considered as $p$ during the execution of the algorithm. Between changing $p_i$ to $p_{i+1}$, we will make at most one call to $\MaxGreedyStep_\sigma$ (by maximality) and at most once call to $\MinGreedyStep_\pi$. Hence $k \le n$ bounds the number of calls to greedy steps by $\Oh(n\log(1/\eps))$.
\end{proof}

\bibliography{frechet}

\end{document}